\documentclass[a4paper,onecolumn,11pt,accepted=2023-09-27]{quantumarticle}
\pdfoutput=1

\usepackage{revsymb, amsmath, amsfonts, amssymb, enumerate, fullpage, amsthm, graphicx, braket, relsize, bbm, mathrsfs, mathtools}
\usepackage[normalem]{ulem}

\usepackage{graphicx,color}
\usepackage{paralist}
\usepackage{subcaption}
\usepackage[labelformat=simple]{subcaption}

\usepackage[numbers,sort&compress]{natbib}
\usepackage{xfrac}

\usepackage[linktocpage=true, colorlinks=true, linkcolor=red, urlcolor=blue, citecolor=blue]{hyperref}

\newtheorem{theo}{Theorem}
\newtheorem{thm}[theo]{Theorem}

\newtheorem{cor}[theo]{Corollary}
\newtheorem{conj}[theo]{Conjecture}

\newtheorem{rem}[theo]{Remark}

\newtheorem{sdp}[theo]{SDP}
\newtheorem{obs}[theo]{Observation}

\newcommand{\cH}{\mathcal{H}}

\newcommand{\I}{\mathcal{I}}
\newcommand{\id}{\mathbb{I}}
\newcommand{\In}{\textbf{I}}
\newcommand{\N}{\textbf{N}}
\newcommand{\Ne}{\mathcal{N}}
\newcommand{\tr}[2]{\mathrm{tr}_{#2} \left\{ #1 \right\}}
\newcommand{\Tr}[1]{\mathrm{tr}\left\{#1 \right\}}
\newcommand{\proj}[1]{\ket{#1}\!\bra{#1}}

\newcommand{\cE}{\xi}

\newcommand{\X}{\mathbb{X}}
\newcommand{\Y}{\mathbb{Y}}
\newcommand{\A}{\mathbb{A}}
\newcommand{\B}{\mathbb{B}}
\newcommand{\As}{\boldsymbol{\Sigma}}

\newcommand{\FA}{\mathsf{R}}

\begin{document}

\title{The resource theory of nonclassicality of channel assemblages}

\author{Beata Zjawin}
\affiliation{International Centre for Theory of Quantum Technologies, University of  Gda{\'n}sk, 80-309 Gda{\'n}sk, Poland}
\email{beata.zjawin@phdstud.edu.ug.pl}
\author{David Schmid}
\affiliation{International Centre for Theory of Quantum Technologies, University of  Gda{\'n}sk, 80-309 Gda{\'n}sk, Poland}
\author{Matty J.~Hoban}
\affiliation{Cambridge Quantum Computing Ltd}
\affiliation{Quantinuum LLC}

\author{Ana Bel\'en Sainz}
\affiliation{International Centre for Theory of Quantum Technologies, University of  Gda{\'n}sk, 80-309 Gda{\'n}sk, Poland}

\maketitle

\begin{abstract}
When two parties, Alice and Bob, share correlated quantum systems and Alice performs local measurements, Alice's updated description of Bob's state can provide evidence of nonclassical correlations. 
This simple scenario, famously introduced by Einstein, Podolsky and Rosen (EPR), can be modified by allowing Bob to also have a classical or quantum system as an input. In this case, Alice updates her knowledge of the {\em channel} (rather than of a state) in Bob's lab. 
In this paper, we provide a unified framework for studying the nonclassicality of various such generalizations of the EPR scenario. 
We do so using a resource theory wherein the free operations are local operations and shared randomness (LOSR). We derive a semidefinite program for studying the pre-order of EPR resources and discover possible conversions between the latter. Moreover, we study conversions between post-quantum resources both analytically and numerically.
\end{abstract}

\newpage
\tableofcontents
\newpage

\section{Introduction}

It is a well-established fact that nature exhibits phenomena that are not explainable by classical laws of physics. Some of these phenomena, such as Bell nonclassicality \cite{Bell64,brunner2014bell}, pertain to correlational aspects of distant physical systems. For instance, Einstein-Podolsky-Rosen (EPR) `steering'~\cite{einstein1935can,schrodinger1935discussion,cavalcanti2009experimental} refers to a scenario where local measurements on half of a system prepared in an entangled state can generate nonclassical correlations between two distant parties. This phenomenon is a crucial resource for various information-theoretic tasks~\cite{wiseman2007steering,uola2020quantum}, such as quantum cryptography~\cite{branciard2012one,gianissecretsharing}, entanglement certification~\cite{cavalcanti2015detection,mattar2017experimental}, randomness certification~\cite{passaro2015optimal,law2014quantum}, and self-testing~\cite{supic2016,goswami2018one,Chen2021robustselftestingof}. The wide applicability of EPR correlations motivates a program of characterizing their resourcefulness~\cite{cavalcanti2009experimental,pusey2013negativity,weight,robustness,gallego2015resource,EPRLOSR}, both in the standard EPR scenario and in multipartite generalizations thereof.

\subsection{Generalizations of the EPR scenario}

The standard EPR scenario consists of two parties, Alice and Bob, that share a bipartite system prepared in an entangled state. By performing measurements on her share of the system, Alice updates her knowledge about the state of Bob's subsystem\footnote{The fact that Alice learns about Bob's distant subsystem is sometimes taken to be evidence for action-at-a-distance. The very term `steering' suggests that Alice has a \textit{causal influence} on Bob's system. In this paper, we do not endorse this view, as we take the causal structure of the EPR scenario to be a \textit{common cause} one. For this reason, we will refer to a `steering scenario' as an {\em EPR scenario} and say that Alice updates her knowledge about the state of Bob’s subsystem, rather than `steering' him. For more discussion of and motivation for this view, see Refs.~\cite{EPRLOSR, cowpie, schmid2020standard, schmid2020type}}. Various generalizations of this standard scenario have been introduced in recent years, wherein Bob {\em also} may probe his system in various ways. In such cases, Alice's measurements allow her to make inferences not merely about the state of Bob's system, but also about the overall process in Bob's laboratory. Instances of such scenarios include the channel EPR scenario~\cite{piani2015channel}, Bob-with-input EPR scenario~\cite{sainz2020bipartite}, measurement device-independent EPR scenario~\cite{cavalcanti2013entanglement}, and the famous Bell scenario~\cite{brunner2014bell,cowpie}. These scenarios are all closely related, and a unified framework for understanding them was introduced in Refs.~\cite{schmid2020type,rosset2020type}. One can easily understand the basic relationship between these distinct scenarios by considering Fig.~\ref{fig:small-resources}. The standard EPR scenario is shown in Fig.~\ref{fig:steering-scenario}; here, Alice and Bob share a quantum system, and Alice performs measurements labeled by classical inputs to obtain classical outputs. When Bob performs measurements with classical input and output systems as well, as illustrated in Fig.~\ref{fig:bell-scenario}, one recovers the Bell scenario~\cite{brunner2014bell,cowpie}. More generally, when Bob's input and output systems are quantum, as shown in Fig.~\ref{fig:channel-scenario}, one obtains the channel EPR scenario~\cite{piani2015channel}. If Bob's input is quantum and the output is classical, one recovers the measurement-device-independent EPR scenario~\cite{cavalcanti2013entanglement}, shown in Fig.~\ref{fig:MDI-scenario}, wherein Alice makes inferences about Bob's measurement channel. Finally, if the input is classical and the output is quantum, one recovers the Bob-with-input EPR scenario~\cite{sainz2020bipartite}, illustrated in Fig.~\ref{fig:BWI-scenario-small}, wherein Alice makes inferences about Bob's state preparation channel. These scenarios all have a similar \textit{common-cause} causal structure; what distinguishes them is the {\em type} -- classical or quantum -- of Bob's input and output system.

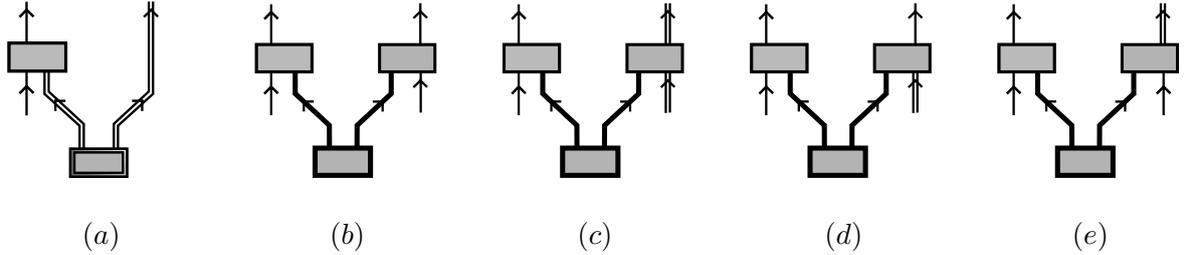
\begin{figure}[h!]
  \begin{center}
  \subcaptionbox{\label{fig:steering-scenario}}
{\put(-35,0){\includegraphics[width=0.125\textwidth]{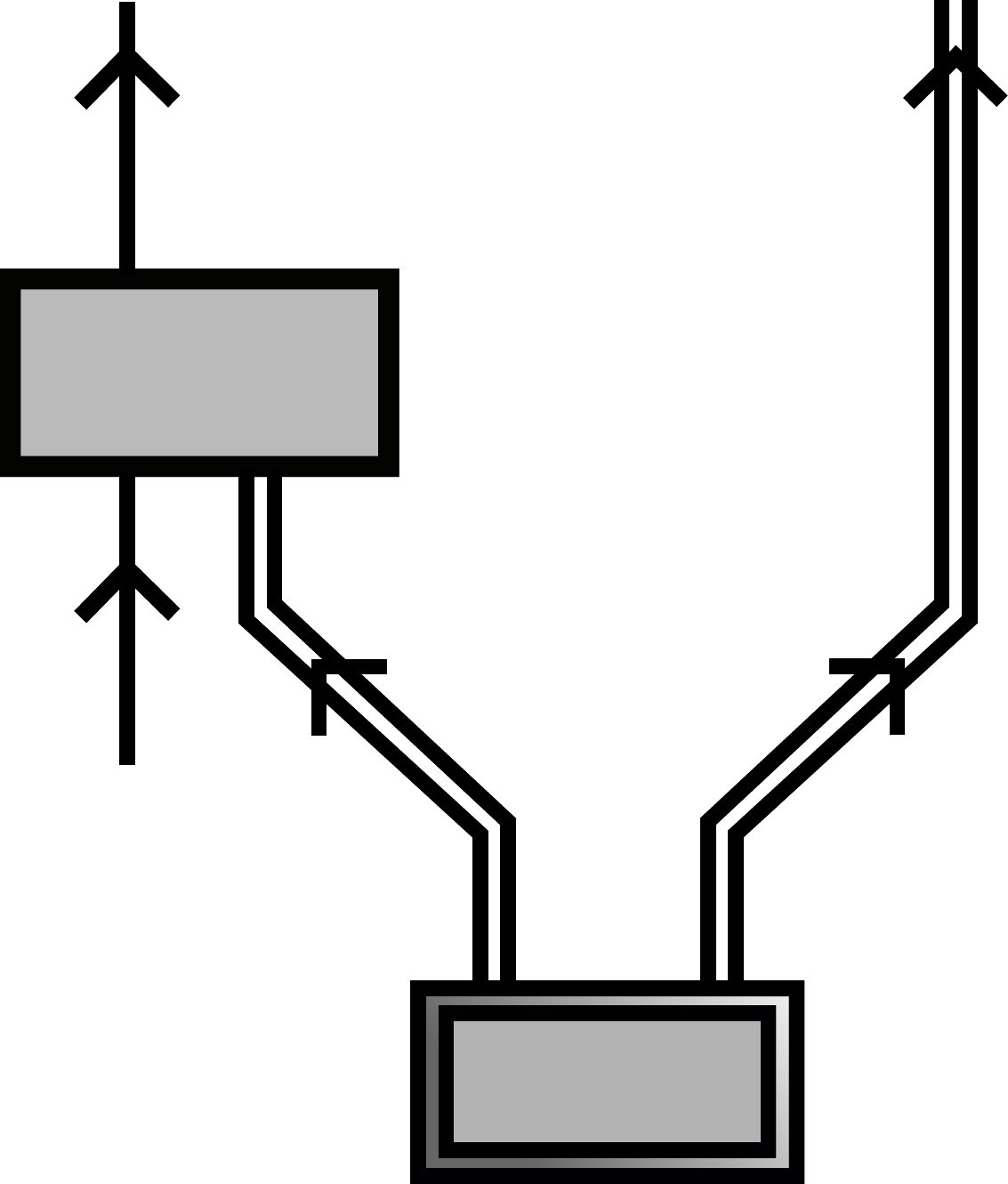}}
\put(-7,-25){$(a)$}
}
\hspace{30mm} 
  \subcaptionbox{\label{fig:bell-scenario}}
{\put(-35,0){\includegraphics[width=0.15\textwidth]{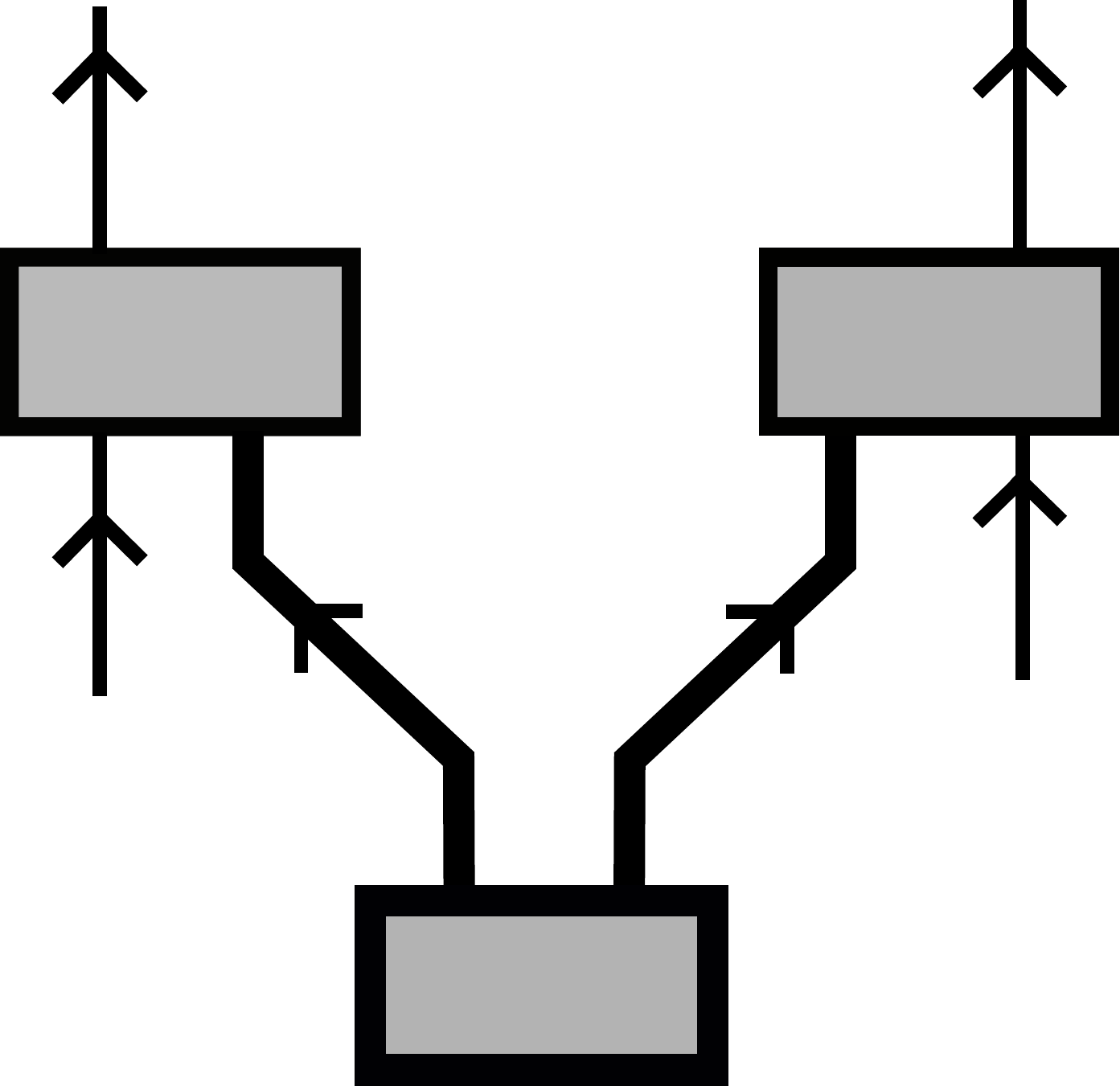}}
\put(-7,-25){$(b)$}
}
\hspace{30mm}
  \subcaptionbox{\label{fig:channel-scenario}}
{\put(-35,0){\includegraphics[width=0.15\textwidth]{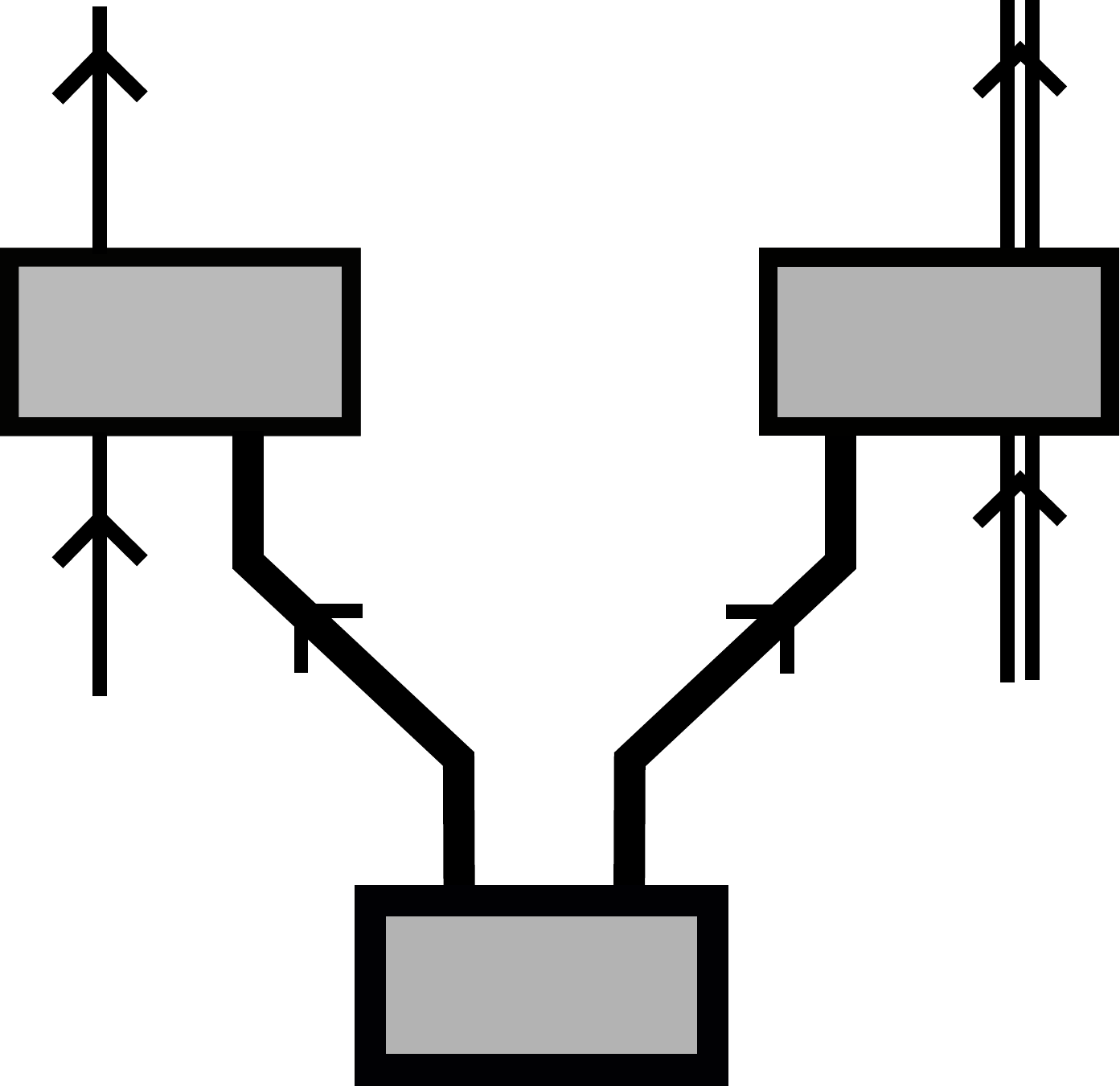}}
\put(-7,-25){$(c)$}
}
\hspace{30mm}
  \subcaptionbox{\label{fig:MDI-scenario}}
{\put(-35,0){\includegraphics[width=0.15\textwidth]{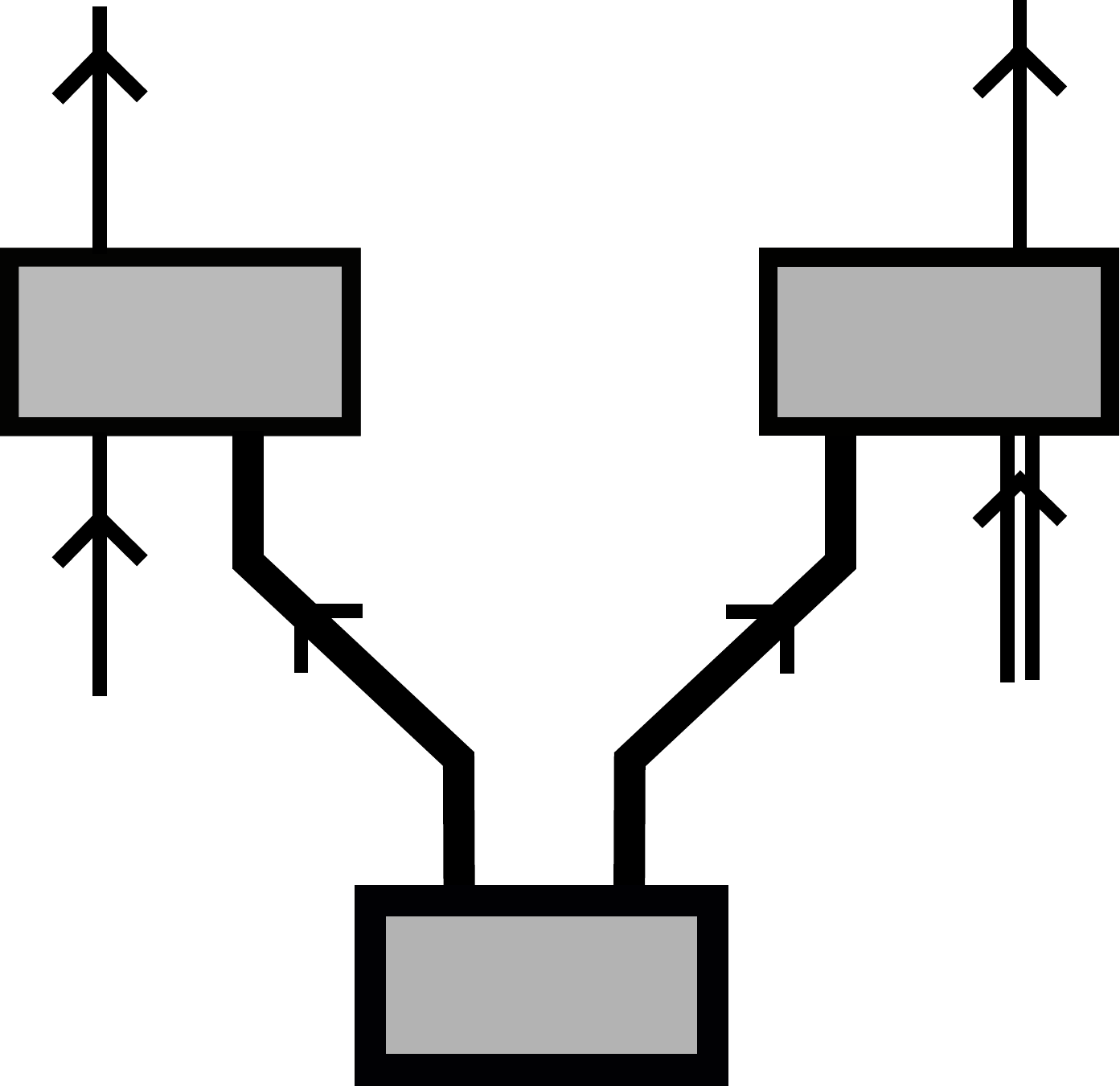}}
\put(-7,-25){$(d)$}
}
\hspace{30mm}
  \subcaptionbox{\label{fig:BWI-scenario-small}}
{\put(-35,0){\includegraphics[width=0.15\textwidth]{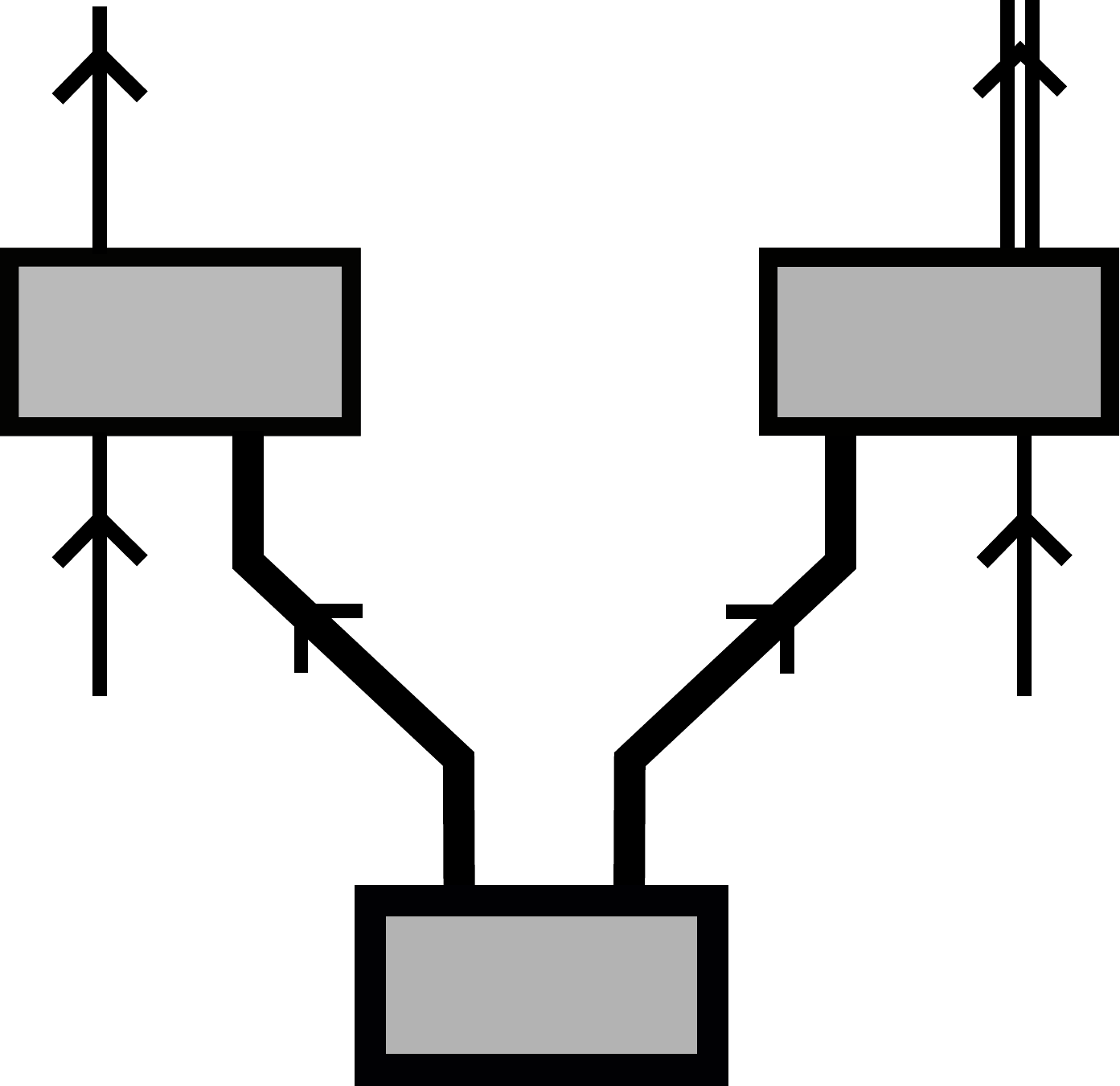}}
\put(-7,-25){$(e)$}
}

  \end{center}
  \vspace*{-10mm}
  \caption{\textbf{Generalizations of the traditional EPR scenario:} (a) traditional bipartite EPR scenario. (b) Bell scenario. (c) Channel EPR scenario. (d) Measurement-device-independent EPR scenario. (e) Bob-with-input EPR scenario. Quantum and classical systems are depicted by double and single lines, respectively. In the special case when the common cause in each scenario is classical, these are all free resources. Thick black lines depict the possibility that the shared systems may be classical, quantum, or even post-quantum, as discussed in Section~\ref{se:theRT}. 
  }\label{fig:small-resources}
\end{figure}

In every type of EPR scenario, Alice chooses from various incompatible methods of refining her knowledge about Bob's process. Her knowledge after performing her local measurements is represented by an {\em ensemble} of Bob's updated processes together with their associated probabilities of arising. In analogy to the traditional EPR scenario, the collection of these ensembles (one for each of her possible measurements) is termed an {\em assemblage}~\cite{pusey2013negativity}. In other words, the standard concept of a collection of ensembles of quantum states in the standard EPR scenario turns into {\em a collection of ensembles of processes} in these generalized EPR scenarios. The assemblage associated to a given scenario contains all the relevant information for characterizing the correlational resources shared between the parties in the scenario. 

\subsection{The resource theory}\label{se:theRT}

In every type of EPR scenario, the most basic problem is identifying which assemblages provide evidence of nonclassicality -- that is, which assemblages can only be generated if Alice and Bob share a nonclassical common cause~\cite{EPRLOSR}. One can then ask about the relative nonclassicality of assemblages, i.e., whether a given assemblage is more or less nonclassical than another. This line of reasoning leads to the framework of resource theories, a principled approach to quantifying the value of a broad range of objects. This approach has proved to be useful in many areas of physics, including coherence~\cite{marvian2016quantify,marvian2016quantum,winter2016operational}, athermality~\cite{brandao2013resource,skrzypczyk2014work,janzing2000thermodynamic,horodecki2013fundamental,gour2015resource,holmes2020quantifying}, LOCC-entanglement~\cite{nielsen1999conditions,bennett1996concentrating,sanders2009necessary}, LOSR-entanglement~\cite{buscemi2012all,schmid2020standard}, and other common-cause processes~\cite{cowpie, schmid2020type,rosset2020type}, including some beyond what quantum theory allows (i.e., \textit{post-quantum} processes) \cite{schmid2021postquantum,barrett2005nonlocal,brunner2009nonlocality}.

In this work, we develop a resource theory which quantifies the nonclassicality of channel assemblages in channel EPR scenarios, subsuming all of the special cases in Fig.~\ref{fig:small-resources}. This resource theory is itself a special case (where one focuses on EPR scenarios) of the type-independent resource theory introduced in Refs.~\cite{schmid2020type,rosset2020type}, which are motivated by the causal modelling perspective on studying nonclassicality~\cite{pearl2000causality,Wood_2015}. Note that the special case of standard EPR scenarios (bipartite and multipartite, with quantum common-cause systems) is analyzed in detail in Ref.~\cite{EPRLOSR}; hence, we will not focus on them here.  Similarly, the special case of Bell scenarios was studied in detail in Ref.~\cite{cowpie}, and so is not studied in detail here. 

In this resource theory of channel assemblages, the \textbf{free resources} -- that is, classical assemblages -- are those that can be generated by local operations and classical common causes. This approach is motivated by previous works~\cite{cowpie,schmid2020standard,schmid2020type,rosset2020type,EPRLOSR} which argue that the resourcefulness in scenarios of this type is best characterized as {\em nonclassicality of the common-cause}. It follows that the operations that can be applied to the resources without increasing their nonclassicality include local operations and classical common causes. In other words, the free operations in our resource theory are {\em local operations and shared randomness} (LOSR). 

Another important ingredient in the definition of a resource theory is its \textbf{enveloping theory}, which
specifies the scope of all possible resources under consideration. Refs.~\cite{gallego2015resource,EPRLOSR} study the case of bipartite and multipartite EPR scenarios where the enveloping theory is taken to be that of quantum resources, i.e., assemblages that Alice and Bob can prepare in the laboratory by performing classically-correlated local actions on a shared quantum system that may be prepared on an entangled state. In this work, however, we consider the enveloping theory to contain all assemblages which can be generated by Alice performing measurements on a bipartite system whose state may be described by an arbitrary theory (a possibility formalized within the framework of generalized probabilistic theories and depicted by thick lines in Fig.~\ref{fig:small-resources}) \cite{cavalcanti2022post}. 
As shown by Ref.~\cite{cavalcanti2022post}, this enveloping theory contains all and only the no-signalling assemblages.
Notably, not every assemblage that satisfies the no-signalling principle can be generated by Alice and Bob sharing quantum resources \cite{sainz2015postquantum,sainz2020bipartite}; this is in analogy to Popescu-Rohrlich boxes \cite{popescu1994quantum} in Bell scenarios. An exception is given by assemblages in traditional bipartite EPR scenarios, where the renowned theorem by Gisin \cite{gisin1989stochastic} and Hughston, Josza, and Wootters \cite{hughston1993complete} (GHJW) proves that the most general assemblages -- the non-signalling ones -- all admit a quantum realization. Our broadening of the enveloping theory of resources allows for a perspective from which to study how quantum resources emerge and how their resourcefulness is bounded relative to all logically possible resources. This is a main advantage of the resource theory presented in this work over the resource theory we developed in Ref.~\cite{EPRLOSR}. For this reason, we mainly focus on studying post-quantum resources in this paper.

\subsection{Summary of main results}

In Section~\ref{sec:channel}, we specify the resource theory of channel EPR scenarios. This section includes the formal definition of the channel EPR scenario and the most general channel assemblage processing under LOSR operations. Moreover, we give a semidefinite program for testing resource conversion under free operations in the corresponding scenario. We then run this program to study the pre-order of channel assemblages. 

In Sections~\ref{sec:BWI} and~\ref{sec:MDI}, we focus on two special cases of channel EPR scenarios: Bob-with-input EPR scenarios and measurement-device-independent EPR scenarios, respectively. We show how the particular set of free operations emerges in each of these scenarios from the general case of Section~\ref{sec:channel} by specifying the necessary system types. We also present simplified semidefinite programs for testing resource conversion under free operations, and discuss properties of the pre-order of resources.

Finally, we discuss related work in Section~\ref{se:relwork}. On the one hand, we note that there are significant conceptual differences between our work and Ref.~\cite{piani2015channel}, which first studied channel EPR scenarios. In particular, Ref.~\cite{piani2015channel} seems to be interested not in nonclassicality in a common-cause scenario, but rather in scenarios with communication between Bob and Alice. While this makes no difference for the set of free resources (when one restricts to no-signaling resources), it does make a significant difference for the overall resource theory. On the other hand, we highlight the similarities and differences between our approach and the resource theory under Local Operations and Shared Entanglement from Ref.~\cite{schmid2021postquantum}.

\section{The channel EPR scenario}\label{sec:channel}

In the channel EPR scenario, two distant parties (Alice and Bob) share a physical system $AB$. In addition, Bob has a quantum system defined on $\mathcal{H}_{B_{in}}$ which he acts on locally by applying a quantum channel $\Lambda^{B_{in} \rightarrow B_{out}}$. As a result, he obtains
a quantum system defined on $\mathcal{H}_{B_{out}}$. 

It is possible that the channel $\Lambda^{B_{in} \rightarrow B_{out}}$ (or its application to system $B_{in}$) is being influenced by the presence of system $B$ in Bob's lab. If this were the case, the effective channel $\Lambda^{B_{in} \rightarrow B_{out}}$ would instead be a larger process $\Gamma^{BB_{in} \rightarrow B_{out}}$ acting on both the quantum system ${B_{in}}$ and the system $B$. If $B$ were a quantum system, $\Gamma^{BB_{in} \rightarrow B_{out}}$ would be called the \emph{channel extension}\footnote{Ref.~\cite{piani2015channel} motivates the correlation between $\Lambda^{B_{in} \rightarrow B_{out}}$ and system $B$ by arguing that the former may be a noisy quantum channel that leaks information to the environment, which can then correlate itself with $B$ or even get all the way to Alice's lab. From our causal perspective, such environment can be taken to be a common cause between Alice and Bob, whereas the leakage of information can be interpreted as the correlation mechanism that arises between Alice's system and Bob's quantum channel through this common cause.} of $\Lambda^{B_{in} \rightarrow B_{out}}$. The idea of a channel EPR scenario is to assume that indeed $\Gamma^{BB_{in} \rightarrow B_{out}}$ is the process happening at Bob's lab, and see what Alice can infer about the quantum channel $\Lambda^{B_{in} \rightarrow B_{out}}$ from the outcome statistics she locally observes on her system $A$. The way Alice probes her system $A$ is by performing measurements on it. Fig.~\ref{fig:channel} illustrates such a channel EPR scenario, where the classical and quantum systems are depicted with single and double wires, respectively, and the thick wires depict systems which may be classical, quantum, or even post-quantum.

Let $x \in \X$ be the classical variable that denotes Alice's choice of measurements, and $a \in \A$ the classical variable that denotes her observed outcome\footnote{In principle, different measurements can have different number of outcomes. Here, for ease of the presentation, we focus on the case where all of Alice's measurements have the same number of outcomes. Our techniques generalize straightforwardly.}. By $p(a|x)$ we denote the probability with which Alice obtains outcome $a$ having performed measurement $x$ on her system $A$. In addition, we denote by $\I_{a|x}(\cdot)$ the instrument that is effectively applied on Bob's quantum system $B_{in}$ to produce a quantum system $B_{out}$, given that Alice has performed measurement $x$ on $A$ and obtained the outcome $a$. It follows then that the object of study in such a channel EPR scenario is the \emph{channel assemblage} of instruments $\In_{\A|\X}= \{\I_{a|x}(\cdot)\}_{a\in\A\,,\, x\in\X}$, with $\tr{\I_{a|x}(\rho)}{B_{out}} = p(a|x)$ for every normalised state $\rho$ of quantum system $B_{in}$, and $\sum_{a\in\A} \I_{a|x}(\cdot) = \Lambda^{B_{in} \rightarrow B_{out}}$ for all $x\in\X$.  

Let us illustrate an example of channel assemblages in the case of Alice and Bob sharing a bipartite quantum system prepared on a state $\rho_{AB}$ -- that is, we take the common cause $AB$ mentioned before to be quantum. In this case, Alice's most general measurements are the \textit{generalized measurements}, i.e., positive operator-valued measures (POVMs), which we denote by $\{M_{a|x}\}_{a\in \A, x \in \X}$. The elements of this channel assemblage will then be:
\begin{equation}\label{eq:channel-ass}
\I_{a|x}(\cdot)=\tr{(M_{a|x} \otimes \id_{B_{out}}) (\id_A \otimes \Gamma^{BB_{in} \rightarrow B_{out}})[\rho_{AB} \otimes (\cdot)]}{A},
\end{equation}
where $(\cdot)$ denotes the input system $B_{in}$. For each $x \in \X$, the instruments $\{\I_{a|x}\}_{a \in \A}$ form a channel which does not depend on $x$, i.e., $\sum_{a\in\A} \I_{a|x}$ is a completely positive and trace preserving (CPTP) map~\cite{NielsenChuang,schmidinitial} which does not depend on Alice's measurement choice. 

In general, however, we will not take the common cause $AB$ to necessarily be a quantum system. This common cause may be classical or even post-quantum. 

\begin{figure}[h!]
  \begin{center}
  \subcaptionbox{\label{fig:channel}}
{\put(-60,40){\includegraphics[width=0.22\textwidth]{Figs/channel.pdf}}
\put(-45,80){$x$}
\put(-45,125){$a$}
\put(-20,150){$\I_{a|x}(\cdot)$ }
\put(38,125){$\cH_{B_{out}}$}
\put(38,80){$\cH_{B_{in}}$}
\put(-20,0){$(a)$}
}
\hspace{80mm} 
\subcaptionbox{\label{fig:channel-LOSR}}
{\put(-85,0){\includegraphics[width=0.3\textwidth]{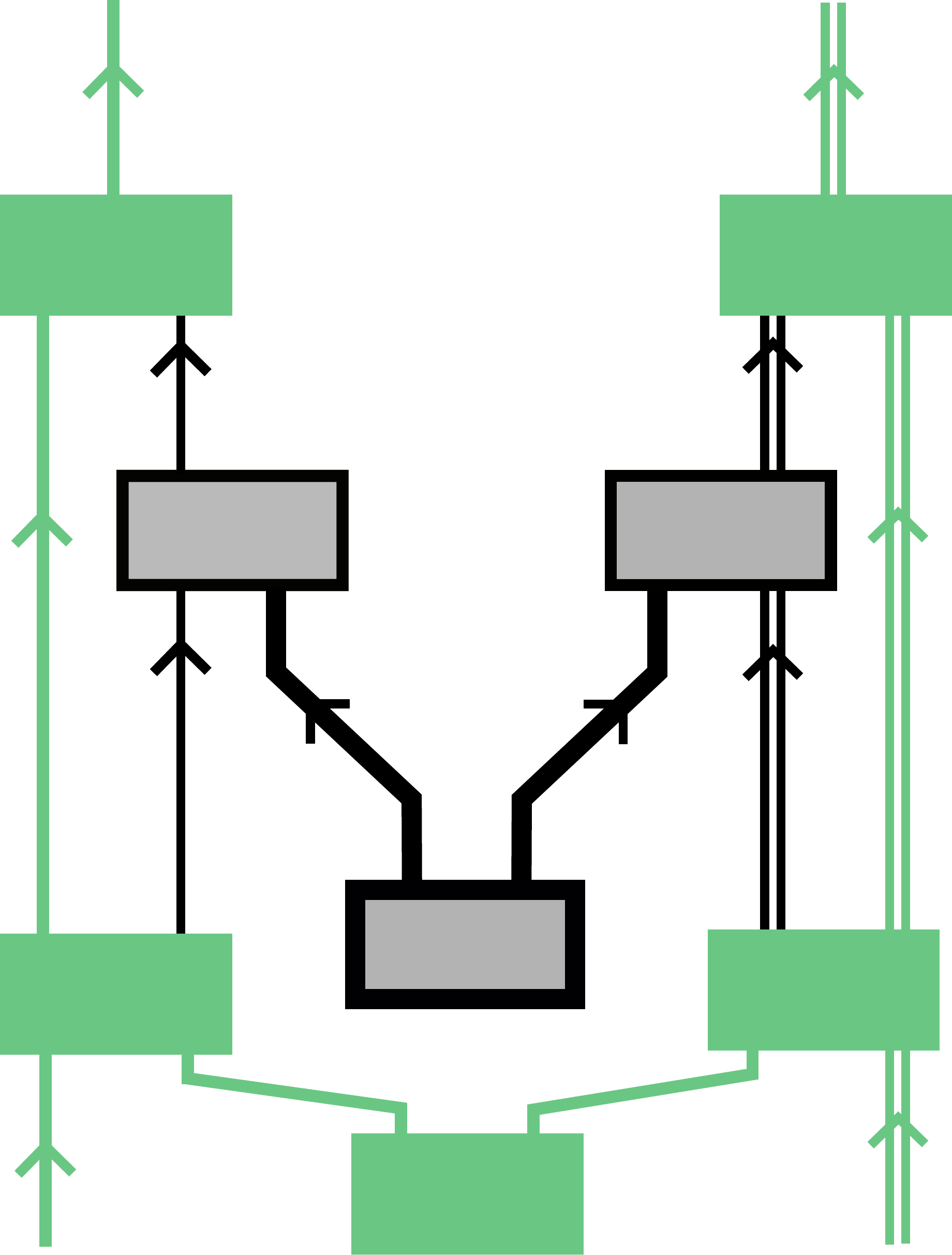}}
\put(-55,75){$x$}
\put(-55,120){$a$}
\put(-5,123){$\cH_{B_{out}}$}
\put(-72,10){$x^{\prime}$}
\put(-60,165){$a^{\prime}$}
\put(-25,180){ $\I'_{a'|x'}(\cdot)$}
\put(42,165){$\cH_{B_{out'}'}$}
\put(52,10){$\cH_{B_{in'}'}$}
\put(0,50){$\cH_{B_{in}}$}
\put(-25,-20){$(b)$}
}

  \end{center}
  \vspace*{-10mm}
  \caption{Depiction of a channel EPR scenario. Arbitrary systems (which may be post-quantum) are represented by thick lines, quantum systems are represented by double lines and classical systems are depicted as single lines. (a) Channel assemblage: Alice and Bob share a (possibly) post-quantum common cause; Alice's input and output systems are classical while Bob's input and output systems are quantum. (b) The most general LOSR operation on an assemblage in a channel EPR scenario. }
\end{figure}

\subsection{LOSR-free channel assemblages}\label{sec:channel-free}

The fundamental objects that define the notion of nonclassicality in a resource theory are the free resources. The free channel assemblages, i.e., assemblages that admit a classical description, are those that can be generated by classical common-cause processes. Therefore, the free assemblages in our framework can be understood as objects that arise when Bob locally applies an instrument and Alice performs a measurement, where both actions depend on a shared classical random variable. In this situation, Alice can refine her description of Bob’s process by learning the classical value of the shared variable. In other words, the classical channel assemblages are those that can be generated from LOSR operations. Formally, the elements of a free channel assemblage can be written as 
\begin{equation}\label{eq:free-channel}
    \I_{a|x}(\cdot)=\sum_{\lambda} p(a|x,\lambda) p(\lambda) \I_{\lambda}(\cdot),
\end{equation}
where $p(\lambda)$ represents the state of the common cause,  $\I_{\lambda}$ is a CPTP map (from system $B_{in}$ to $B_{out}$) for each $\lambda$, and  $p(a|x,\lambda)$ is a valid conditional probability distributions for all values of $\lambda$. 

A convenient representation of Eq.~\eqref{eq:free-channel} is given by the Choi-Jamiołkowski isomorphism \cite{choi1975completely,jamiolkowski1972linear}. Recall that every CPTP map $\mathcal{E}: \mathcal{H}_B \rightarrow \mathcal{H}_{B'}$ can be associated with an operator $W$ on $\mathcal{H}_B \otimes \mathcal{H}_{B'}$ such that $\mathcal{E}(\rho_B)=d_B\, \tr{W\,(\,\id_{B^{\prime}} \otimes \rho_{B}^{T})}{B}$, where $d_B$ is the dimension of the system $\rho_B$. Conversely, the operator $W$ can be written as $W=(\mathcal{E}\otimes\id_{B'})\ket{\Omega}\bra{\Omega}$, with $\ket{\Omega}=\frac{1}{\sqrt{d_B}}\sum_{i=1}^{d_B} \ket{ii}$. We will use this representation throughout the paper. Let $J_{a|x}$ and $J'_{\lambda}$ represent the Choi-Jamiołkowski operators that correspond to $\I_{a|x}$ and $\I_{\lambda}$, respectively. Moreover, denote $p(\lambda)J'_{\lambda}$ as $J_{\lambda}$. If $\{\I_{a|x}\}_{a,x}$ form an LOSR-free channel assemblage, each $J_{a|x}$ can be expressed as
\begin{equation}\label{eq:free-channel-Choi}
    J_{a|x}=\sum_{\lambda} p(a|x,\lambda) J_{\lambda}.
\end{equation}
Therefore, checking whether an assemblage $\In_{\A|\X}$ is free amounts to checking if the operators $J_{a|x}$ admit a decomposition of the form given in Eq.~\eqref{eq:free-channel-Choi}. Moreover, each probability distribution $p(a|x,\lambda)$ can be decomposed as $p(a|x,\lambda)=\sum_{\lambda'}p(\lambda'|\lambda) D(a|x,\lambda')$, where $D(a|x,\lambda')$ is a deterministic probability distribution. It follows that testing whether a channel assemblage $\In_{\A|\X}$ is free requires a single instance of a semidefinite program (SDP):

\begin{sdp}\label{SDP-channel-free}
The channel assemblage $\In_{\A|\X}$ is LOSR-free if and only if the following SDP is feasible:
\begin{align}
\begin{split}
\textrm{given} \;\;\;& \{ J_{a|x}\}_{a,x}\,,\;  \{D(a|x,\lambda)\}_{\lambda,a,x} \,, \;  \\
    \textrm{find} \;\;\;& \{(J_{\lambda})_{B_{in}B_{out}}\}_{\lambda}  \\
    \textrm{s.t.} \;\;\;& \begin{cases} J_{\lambda} \geq 0\qquad \forall \lambda\,,\\
      \tr{J_{\lambda}}{B_{out}} \propto \id_{B_{in}} \qquad \forall \lambda\,, \\
     \sum_{\lambda} \tr{J_{\lambda}}{B_{out}} = \frac{1}{d_{B_{in}}}\, \id_{B_{in}}\,, \\
      J_{a|x}=\sum_{\lambda} D(a|x,\lambda)  J_{\lambda} \quad \forall \, a\in\A\,,\,x\in\X    \,.
      \end{cases}
    \end{split}
\end{align}
\end{sdp}

\noindent
Here $d_{B_{in}}$ is the dimension of the Hilbert space $\mathcal{H}_{B_{in}}$ associated with system $B_{in}$. From here on, as above, we will use the notation $d_{X}$ to refer to the dimension of the Hilbert space associated with system $X$.

\subsection{LOSR transformations between channel assemblages}

The most general LOSR transformation of a channel EPR assemblage is presented in Fig.~\ref{fig:channel-LOSR}. Formally, this LOSR operation transforms elements of a channel assemblage $\In_{\A|\X}$ into new ones that form $\In'_{\A'|\X'}$ as follows:
\begin{align}\label{eq:LOSRtrans-channel} 
\begin{split}
\I'_{a'|x'}(\cdot)=\sum_{\lambda}  \sum_{a,x} 
p(\lambda)\,p(a^{\prime},x|a,x^{\prime},\lambda)\,  \Lambda_\lambda^{B_{out}S \rightarrow B_{out'}'} \,(\I_{a|x}\otimes \id_S)\, \Lambda_{\lambda}^
{B_{in'}' \rightarrow B_{in}S} (\cdot)\,,
\end{split}
\end{align}
where 
\begin{compactitem}
\item $p(\lambda)$ is the probability distribution over the classical random variable $\lambda$ which coordinates all the local transformations.

\item $p(a^{\prime},x|a,x^{\prime},\lambda)$ encodes the classical pre- and post-processing of Alice's input $x$ and output $a$. The output of this process is Alice's new outcome $a^\prime$. The probability distribution $p(a^{\prime},x|a,x^{\prime},\lambda)$ satisfies the no-retrocausation condition (the variable $a$ cannot influence the value of the variable $x$, i.e., $p(x|a,x^{\prime},\widetilde{\lambda})=p(x|x^{\prime},\widetilde{\lambda})$).

\item $\Lambda_{\lambda}^{B_{in'}' \rightarrow B_{in}S}$ is the map corresponding to Bob's local pre-processing. It acts on his quantum input on $\cH_{B_{in'}'}$ and outputs a quantum system on a Hilbert space $\cH_{B_{in}} \otimes \cH_{S}$ (the index $S$ corresponds to the quantum side channel of the local processing\footnote{It is worth noticing that the quantum system on $\cH_{S}$ is not of arbitrary dimension. By the results of Ref.~\cite{gutoski2007toward}, its dimension is restricted by the product of the dimensions of $\cH_{B_{in'}'}$ and $\I'_{a'|x'}(\cdot)$.}). 

\item $\Lambda_\lambda^{B_{out}S \rightarrow B_{out'}'}$ is the map corresponding to Bob's post-processing stage. The output of this process is Bob's new quantum system on a Hilbert space $\cH_{B_{out'}'}$.
\end{compactitem}

It is convenient to express Bob's pre- and post-processing as a single completely positive and trace non-increasing (CPTNI) map, which we denote by $\cE_{\lambda}^{\,B_{in'}'B_{out} \rightarrow B_{in}B_{out'}'}$, where $\sum_{\lambda} \cE_{\lambda}^{\,B_{in'}'B_{out} \rightarrow B_{in}B_{out'}'}$ forms a CPTP map. The inputs of this new map are the inputs of Bob's pre- and post-processing (defined on $\cH_{B_{in'}'} \otimes \cH_{B_{out}}$) and the outputs of this new map are the outputs of Bob's pre- and post-processing (defined on $\cH_{B_{out'}'} \otimes \cH_{B_{in}}$). The map $\cE_{\lambda}^{\,B_{in'}'B_{out} \rightarrow B_{in}B_{out'}'}$ must be such that there is no signalling from the system defined on $\cH_{B_{out}}$ to the system defined on $\cH_{B_{in}}$, that is, the input of Bob's post-processing can not influence the output of his pre-processing. This condition can be expressed as follows:
\begin{align}\label{eq--} 
\forall \, \lambda \quad \exists \,\, F_{\lambda}^{B_{in'}' \rightarrow B_{in}} \qquad \textrm{s.t.} \qquad \tr{\cE_{\lambda}^{\,B_{in'}'B_{out} \rightarrow B_{in}B_{out'}'}} {B_{out'}'} = F_{\lambda}^{B_{in'}' \rightarrow B_{in}} \otimes \id_{B_{out}}\,,
\end{align}
where $\sum_{\lambda} F_{\lambda}^{B_{in'}' \rightarrow B_{in}}$ is a CPTP map. Moreover, it is convenient to use the Choi-Jamiołkowski isomorphism to express the maps of interest. Let $J'_{a'|x'}$,  $J_{a|x}$ and $J_{\cE\, \lambda}$ correspond to the Choi representation of maps $\I'_{a'|x'}$, $\I_{a|x}$ and $\cE_{\lambda}^{B_{in'}'B_{out} \rightarrow B_{in}B_{out'}'}$, respectively. Then, Eq.~\eqref{eq:LOSRtrans-channel} can be expressed as
\begin{align}\label{eq:LOSR-for-SDP} 
\begin{split}
J'_{a'|x'}=\sum_{\lambda}  \sum_{a,x} 
D(a^{\prime}|a,x^{\prime},\lambda)\,D(x|x^{\prime},\lambda)\,  J_{a|x}\, * J_{\cE\, \lambda} \,,
\end{split}
\end{align}
where the maps are composed using the link product $*$ described in Ref. \cite{chiribella2009theoretical}:
\begin{align}\label{eq---} 
J_{a|x}\, * J_{\cE\, \lambda} = d_{B_{in}} \, d_{B_{out}} \,\tr{ ( \id_{B_{in'}'B_{out'}'} \otimes J_{a|x}^{T_{B_{out}}}) \, J_{\cE\, \lambda}^{T_{B_{in}}}}{B_{in}B_{out}}\,.
\end{align}
Notice that Alice's local processing is expressed in terms of deterministic probability distributions $D(\cdot)$ -- this representation was discussed in detail in previous works~\cite{FinePRL,cowpie,EPRLOSR}; we recall this discussion in Appendix~\ref{app:deterministic}. In this scenario,  the total number of the deterministic strategies encoded in $\lambda$ is equal to $|\A^{\prime}|^{|\A|\times|\X'|}\times|\X|^{|\X'|}$.

An assemblage $\In_{\A|\X}$ can be converted therefore into a different assemblage $\In'_{\A^{\prime}|\X^{\prime}}$ under LOSR if and only if there exist a collection of Choi states $\{J_{\cE\,\lambda}\}$ and operators $\{F_{\lambda}\}$  such that the elements of $\In'_{\A^{\prime}|\X^{\prime}}$ can be decomposed as in Eq.~\eqref{eq:LOSR-for-SDP} with the conditions of Eq.~\eqref{eq--}. This correspondence enables us to derive the following SDP that checks whether such a decomposition is possible.

\begin{sdp}\label{SDP-channel}
The channel assemblage $\In_{\A|\X}$ can be converted to the channel assemblage $\In'_{\A'|\X'}$ under LOSR operations, denoted by $\In_{\A|\X} \, \overset{\text{LOSR}}{\longrightarrow} \, \In'_{\A^{\prime}|\X^{\prime}}$, if and only if the following SDP is feasible:
\begin{align}
\begin{split}
\textrm{given} \;\;\;& \{ J_{a|x}\}_{a,x}\,,\; \{J'_{a'|x'}\}_{a',x'}\,,\; \{D(a^{\prime}|a,x^{\prime},\lambda)\}_{\lambda,a^{\prime},a,x^{\prime}} \,, \;  \{D(x|x^{\prime},\lambda)\}_{\lambda,x,x^{\prime}} \\
    \textrm{find} \;\;\;& \{(J_{\cE\,\lambda})_{B_{in}B_{in'}'B_{out}B_{out'}'}\}_{\lambda} \,,\;\{(J_{F\,\lambda})_{B_{in}B_{in'}'}\}_{\lambda}  \\
    \textrm{s.t.} \;\;\;& \begin{cases} 
   J_{\cE\,\lambda} \geq 0 \quad \forall \lambda\,,\\
      \tr{J_{\cE\,\lambda}}{B_{out'}'B_{in}} \propto \id_{B_{out}B_{in'}'} \quad \forall \lambda\,, \\
      \sum_{\lambda} \tr{J_{\cE\,\lambda}}{B_{out'}'B_{in}} = \frac{1}{d_{B_{out}}d_{B_{in'}'}}\, \id_{B_{out}B_{in'}'}\,, \\
          J_{F\,\lambda} \geq 0 \quad \forall \lambda\,,\\
      \tr{J_{F\,\lambda}}{B_{in}} \propto \id_{B_{in'}'} \;\;\;\; \forall \lambda\,, \\
      \sum_{\lambda} \tr{J_{F\,\lambda}}{B_{in}} = \frac{1}{d_{B_{in'}'}}\, \id_{B_{in'}'}\,, \\
      \tr{J_{\cE\,\lambda}}{B_{out'}'} = J_{F\,\lambda} \otimes \frac{1}{d_{B_{out}}}\,\id_{B_{out}} \quad \forall \lambda \,, \\
            J'_{a'|x'}=\sum_{\lambda}  \sum_{a,x} 
D(a^{\prime}|a,x^{\prime},\lambda)\,D(x|x^{\prime},\lambda)\,   J_{a|x}\, * J_{\cE\,\lambda} \,.
      \end{cases}
    \end{split}
\end{align}
When the conversion is not possible, we denote it by $\In_{\A|\X} \, \overset{\text{LOSR}}{\not\longrightarrow} \, \In'_{\A^{\prime}|\X^{\prime}}$.
\end{sdp}

This SDP is a feasibility problem. Let us mention that a feasibility problem can always be converted into an optimization of some linear objective function. Such formulation of the program is favourable if one wants to increase the robustness of the results. In Appendix~\ref{app:SDP}, we provide an alternative formulation of SDP~\ref{sec:channel} in this form. 

\subsection{Properties of the pre-order}\label{sec:preorder-channel}

One channel assemblage is claimed to be at least as nonclassical as another if a conversion under free operations from the former to the latter is possible. Therefore, testing assemblage conversion under LOSR operations provides information about relative nonclassicality of assemblages. We now proceed to study the conversions between quantum-realizable and post-quantum channel assemblages defined below.

\subsubsection{Quantum-realizable channel assemblages}

First, we focus on an infinite family of quantum-realizable channel assemblages. Consider a channel assemblage with $\A=\X=\{0,1\}$. Assume that Alice and Bob have access to a two qubit system in a Bell state, denoted by $\rho_{AB} = \ket{\phi}\bra{\phi}$, with $\ket{\phi} =  \frac{\ket{00}+\ket{11}}{\sqrt{2}}$. Define a family of channel assemblages $\FA$ as follows: for an angle $\theta \in \left(0,\sfrac{\pi}{2}\right]$, Bob performs a controlled rotation around the $y$ axis on systems $BB_{in}$, where the control qubit is $B$ (it is entangled with Alice's qubit $A$) and $\theta$ is the rotation angle. We denote such quantum gate by $\text{CR}_{y}^{\theta}$; it can be written as:

\begin{align}\label{eq:CR}
\text{CR}_{y}^{\theta} = 
  \begin{pmatrix}
    1 & 0 & 0 & 0\\
    0 & 1 & 0 & 0\\
    0 & 0 & \cos(\theta/2) & -\sin(\theta/2)\\
    0 & 0 & \sin(\theta/2) & \cos(\theta/2)
  \end{pmatrix} .
\end{align}
The measurements Alice performs on her share of the system are given by ${M}_{a|0} = \frac{1}{2}\{\id + (-1)^a \sigma_z\}\,$ and ${M}_{a|1} = \frac{1}{2}\{\id + (-1)^a \sigma_x\}$, where $\sigma_z$ and $\sigma_x$ are Pauli matrices. Then, the elements of the family $\FA$ are the following: 

\begin{align}\label{eq:CRfam}
\FA &= \left\{ \In^{\,\theta}_{\A|\X} \, \Big\vert \, \theta \in \left(0,\sfrac{\pi}{2}\right] \right\}\,,\\ \nonumber
\text{where} \quad \In^{\,\theta}_{\A|\X} &= \left\{\I_{a|x}^{\theta}(\cdot) \right\}_{a\in \A, x \in \X} \,,\\
\text{with} \quad \I_{a|x}^{\theta}(\cdot) &= \tr{(M_{a|x} \otimes \id_{B_{out}}) \, \tr{(\id_A \otimes \text{CR}_{y}^{\theta}) \,(\rho_{AB} \otimes (\cdot)_{B_{in}})(\id_A \otimes \text{CR}_{y}^{\theta})^{\dagger}}{B}}{A} \,, \nonumber 
\end{align}
The family of assemblages $\FA$ is illustrated in Fig.~\ref{fig:CRfam}. 

\begin{figure}[h!]
  \begin{center}
\includegraphics[width=0.32\textwidth]{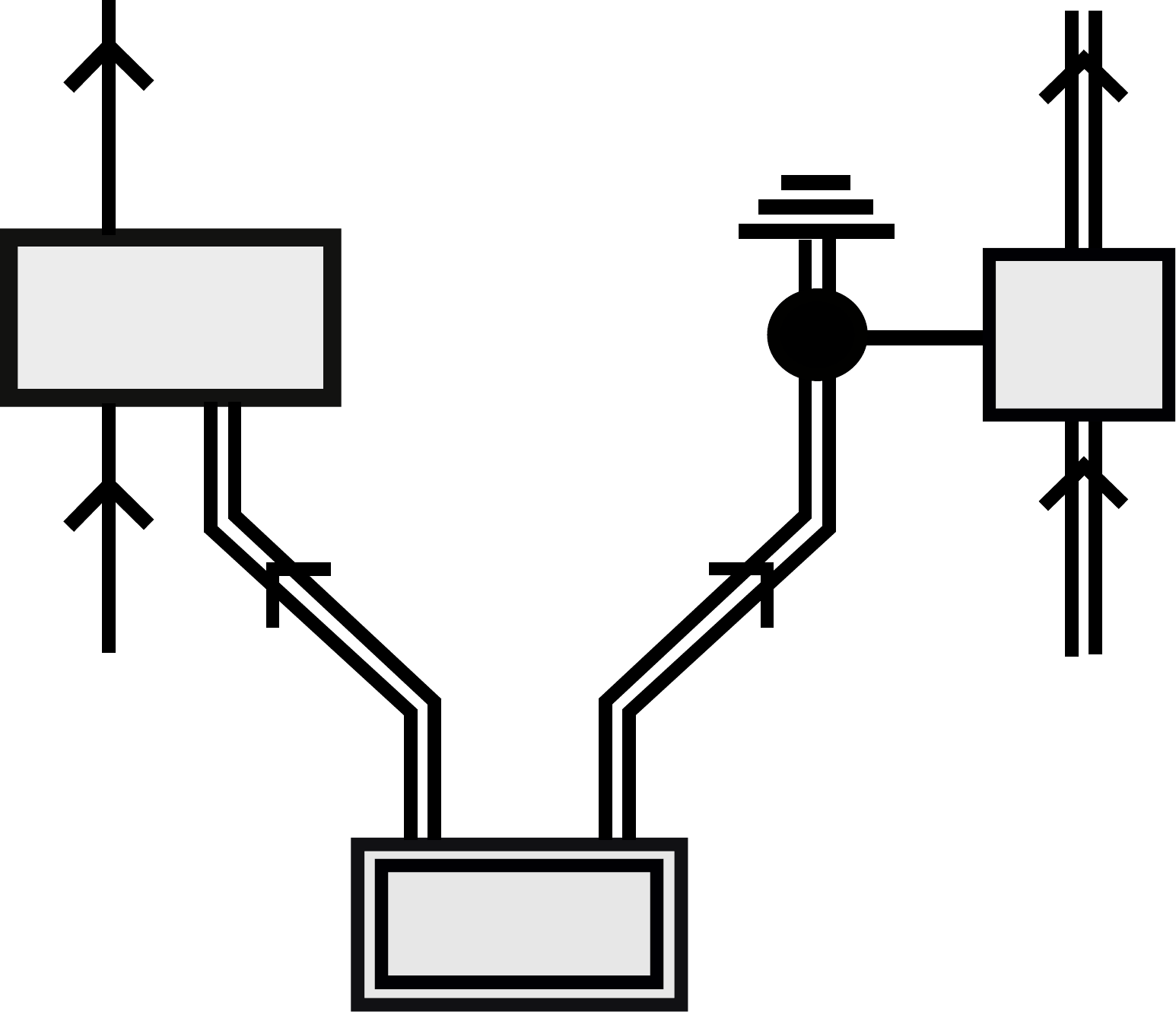}
\put(-150,55){$x$}
\put(-150,110){$a$}
\put(-20,80){$\text{R}_{y}^{\theta}$}
\put(0,55){$B_{in}$}
\put(0,110){$B_{out}$}
\put(-89,9){$\rho_{AB}$}

  \end{center}
  \caption{Depiction of an assemblage belonging to the family $\FA$. Quantum systems are represented by double lines and classical systems are depicted as single lines. Alice and Bob share an entangled common-cause $\rho_{AB}$. Depending on the state of the control qubit, Bob either does nothing or implements a rotation by angle $\theta$ around the $y$ axis on his input system. }\label{fig:CRfam}
\end{figure}
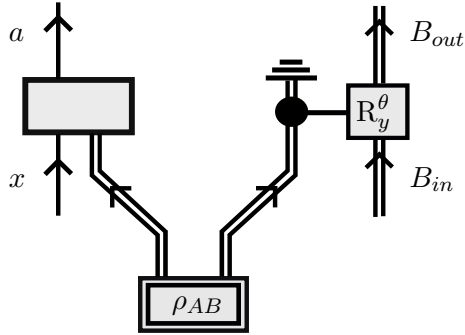

The family $\FA$ has an infinite number of elements indexed by the angle $\{\theta\}$. To study the pre-order of assemblages in this family, we convert the assemblages into Choi form and run the SDP~\ref{SDP-channel} (in Matlab \cite{MATLAB:2010}, using the software CVX \cite{grant2013cvx,blondel2008recent}, the solver SeDuMi \cite{sturm1999using} and the toolbox QETLAB \cite{qetlab}; see the code at \cite{github}). By checking what the possible conversions are between resources characterized by different values of $\theta$, we observe
that a conversion $\In^{\,\theta_1}_{\A|\X} \overset{\text{LOSR}}{\longrightarrow} \In^{\,\theta_2}_{\A|\X}$ is possible only when $\theta_1 > \theta_2$ for every pair $\{\theta_1,\theta_2\}$ that we checked. Based on this observation, we formulate the following conjecture:
\begin{conj}[]{}\label{conj}
Let the two resources $\In^{\,\theta_1}_{\A|\X}$ and $\In^{\,\theta_2}_{\A|\X}$ belong to the infinite family of assemblages $\FA$. Then, the conversion $\In^{\,\theta_1}_{\A|\X} \overset{\text{LOSR}}{\longrightarrow} \In^{\,\theta_2}_{\A|\X}$ is possible if and only if $\theta_1 \geq \theta_2$.
\end{conj}

We run SDP~\ref{SDP-channel} to study
two more infinite families of assemblages analogous to $\FA$ with the exception that rotations around the $z$-axis (the $x$-axis for the second family) instead of around the $y$-axis are implemented. For both cases, we observe a behaviour analogous to that stated in Conjecture~\ref{conj}. Finally, by checking the conversions between assemblages belonging to different families, we find assemblages that are interconvertible (See Fig.~\ref{fig:SDP-results}), which leads us to the following conjecture:
\begin{conj}[]{}\label{conj2}
Let $\In^{\,\theta\,, i}_{\A|\X}$ denote an assemblage generated when Bob applies $\text{CR}_{i}^{\theta}$. Then, for a fixed value of $\theta$, assemblages $\{\In^{\,\theta\,, x}_{\A|\X},\In^{\,\theta\,, y}_{\A|\X},\In^{\,\theta\,, z}_{\A|\X}\}$ are in the same LOSR equivalence class. 
\end{conj}
Our explorations are summarized in Fig.~\ref{fig:SDP-results}.

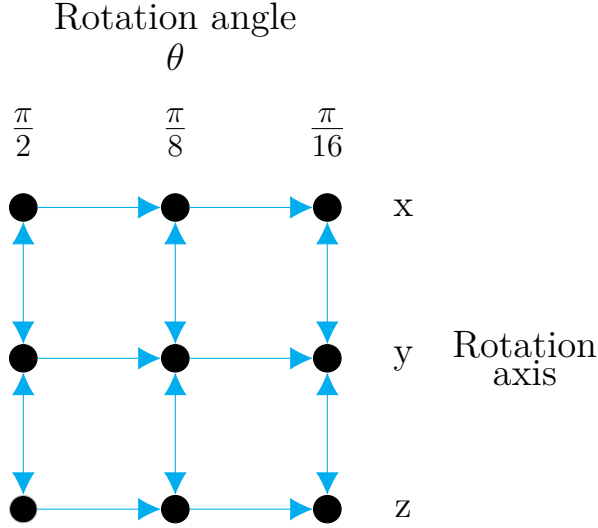
\begin{figure}
 \centering
  \usetikzlibrary{arrows.meta}
\colorlet{lgray}{gray!50}
\tikzset{>={Latex[width=3mm,length=3mm]}}

\begin{tikzpicture}
\foreach \x in {0,1,2}{
  \foreach \y in {1,2}{
      \node[circle,fill=black,draw] (\x;\y)  at (2*\x,2*\y)  {} ;
    }
} 

\foreach \x in {1,2}{
  \foreach \y in {0}{
      \node[circle,fill=black,draw] (\x;\y)  at (2*\x,2*\y)  {} ;
    }
}

\node[circle,draw=lgray,fill=black,draw] (0;0) at (0,0) {} ;

\foreach \i in {0,1,2}{
      {\draw[ <->,draw=cyan] (\i;1) -- (\i;0);}
      {\draw[<->,draw=cyan] (\i;2) -- (\i;1);}
    }
    
\foreach \i in {0,1,2}{
    {\draw[->,draw=cyan] (0;\i) -- (1;\i);}
    {\draw[->,draw=cyan] (1;\i) -- (2;\i);}
}

\node[] at (2*3.3,2*1.1) {\Large Rotation};
\node[] at (2*3.3,2*0.9) {\Large axis};
\node[] at (2*2.5,2*0) {\Large z};
\node[] at (2*2.5,2*1) {\Large y};
\node[] at (2*2.5,2*2) {\Large x};

\node[] at (2*1,2*3.25) {\Large Rotation angle};
\node[] at (2*1,2*3) {\Large $\theta$};
\node[] at (2*0,2*2.5) {\LARGE $\frac{\pi}{2}$};
\node[] at (2*1,2*2.5) {\LARGE $\frac{\pi}{8}$};
\node[] at (2*2,2*2.5) {\LARGE $\frac{\pi}{16}$};

\end{tikzpicture}
  \caption{Possible conversions between elements of $\FA$. The black dots represent the assemblages $\In^{\,\theta\,, i}_{\A|\X}$, where $i\in\{x,y,z\}$. The arrows represent possible conversions. }
  \label{fig:SDP-results}
\end{figure}

\subsubsection{Post-quantum channel assemblages} 

Studying possible conversions among channel assemblages may also give us insight into the pre-order of post-quantum resources. We now focus on a channel EPR scenario where $\X=\{0,1,2\}$, $\A=\{0,1\}$, and $\cH_{B_{out}} = \cH_{B}$ are qubit Hilbert spaces, and we consider conversions between two post-quantum channel assemblages defined below\footnote{These two post-quantum channel assemblages are generalizations of Bob-with-input assemblages introduced in Ref. \cite{sainz2020bipartite}. We elaborate on this point later in the text.}.

Our first example of a post-quantum channel assemblage, depicted in Fig.~\ref{fig:channel-PTP}, can be conveniently expressed in a  mathematical way as:
\begin{align}\label{eq:PTP-channel}
\In^{PTP}_{\A|\X} &= \left\{\I^{PTP}_{a|x}\right\}_{a\in \A, \, x \in \X} \,,\\
\text{with} \quad & \begin{cases} \I_{a|x}^{PTP}(\cdot) = \tr{(M_{a|x} \otimes \id_{B_{out}}) \, \tr{(\id_A \otimes \text{CT}^{BB_{in} \rightarrow B_{out}}) \,(\rho_{AB} \otimes (\cdot)_{B_{in}})}{B}}{A} \,, \nonumber \\
{M}_{a|0} = \frac{\id + (-1)^a \sigma_x}{2}\,,\quad {M}_{a|1} = \frac{\id + (-1)^a \sigma_y}{2}\,,\quad
{M}_{a|2} = \frac{\id + (-1)^a \sigma_z}{2}\,, \nonumber \end{cases}
\end{align}
where $\rho_{AB}=\ket{\phi}\bra{\phi}$, with $\ket{\phi} = \frac{\ket{00}+\ket{11}}{\sqrt{2}}$, and $\sigma_x$, $\sigma_y$, and $\sigma_z$ are the Pauli operators. The operation CT$^{BB_{in} \rightarrow B_{out}}$ is a controlled-transpose operation, where the control system is $B_{in}$, and the transpose is applied on the system $B$. For simplicity, in this section we denote $\In^{PTP}_{\A|\X}=\In^{PTP}$. Here, the abbreviation PTP stands for positive trace preserving.

It is important to note that the expression given in Eq.~\eqref{eq:PTP-channel} and the assemblage preparation procedure illustrated in Fig.~\ref{fig:channel-PTP} give a convenient representation of $\In^{PTP}$, but is not meant to be taken as the unique description of this assemblage. The use of the transpose map here, which is positive but not completely positive, is just one of the possible mathematical ways to represent this post-quantum assemblage. 

One then may argue that the channel assemblage $\In^{PTP}$ is a post-quantum assemblage. To see this, one can simply see that a special case of $\In^{PTP}$ -- that where Bob's input states are classical labels encoded into an orthonormal basis -- was shown to be a post-quantum assemblage \cite{sainz2020bipartite}. It follows hence that $\In^{PTP}$ is post-quantum as well.

\begin{figure}[h!]
  \begin{center}
  \subcaptionbox{\label{fig:channel-PTP}}
{\put(-60,20){\includegraphics[width=0.3\textwidth]{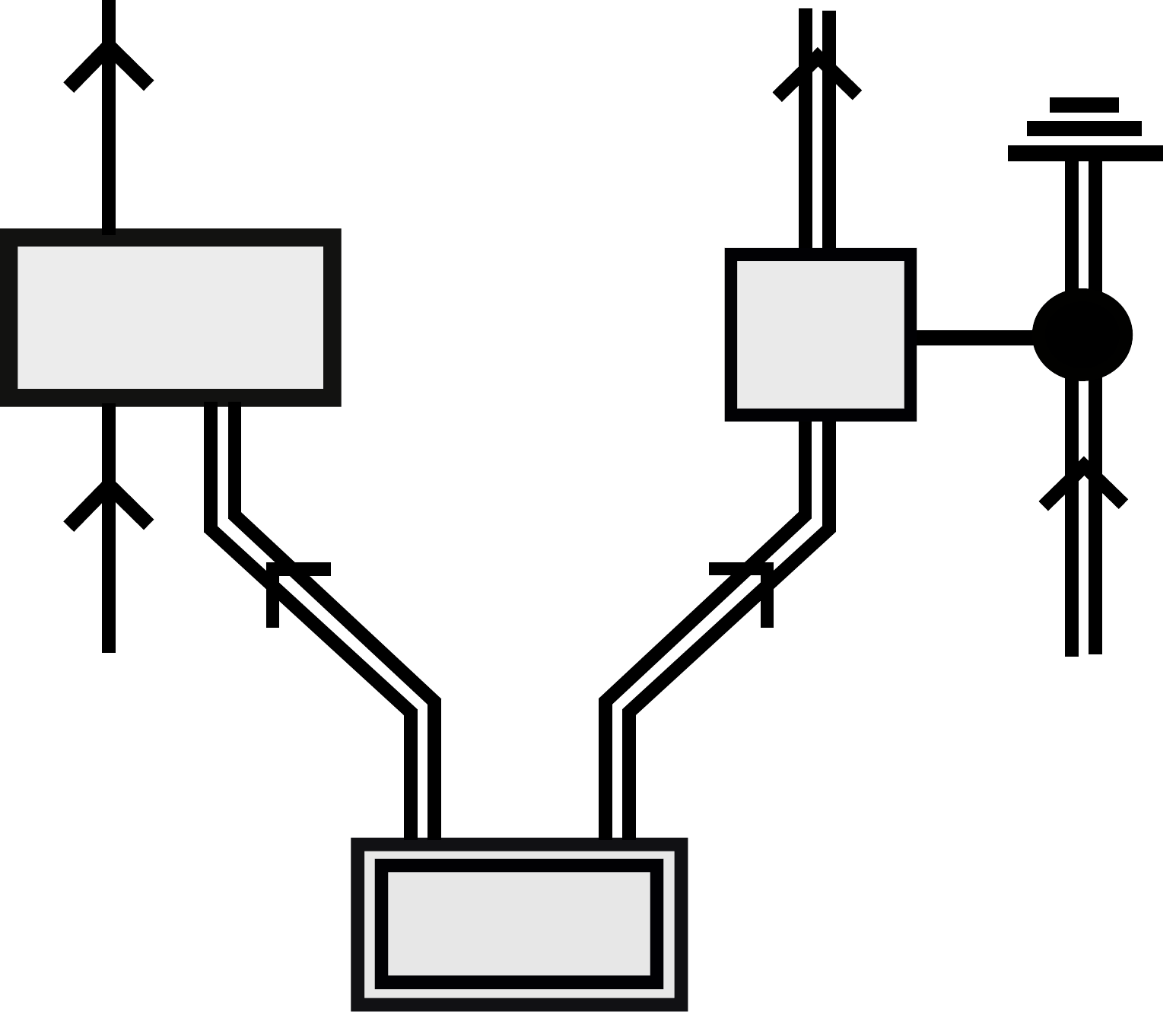}}
\put(-65,70){$x$}
\put(-65,125){$a$}
\put(30,95){T}
\put(75,70){$B_{in}$}
\put(5,135){$B_{out}$}
\put(-6,27){$\rho_{AB}$}
\put(-6,0){$(a)$}
}
\hspace{80mm} 
\subcaptionbox{\label{fig:channel-PR}}
{\put(-80,0){\includegraphics[width=0.3\textwidth]{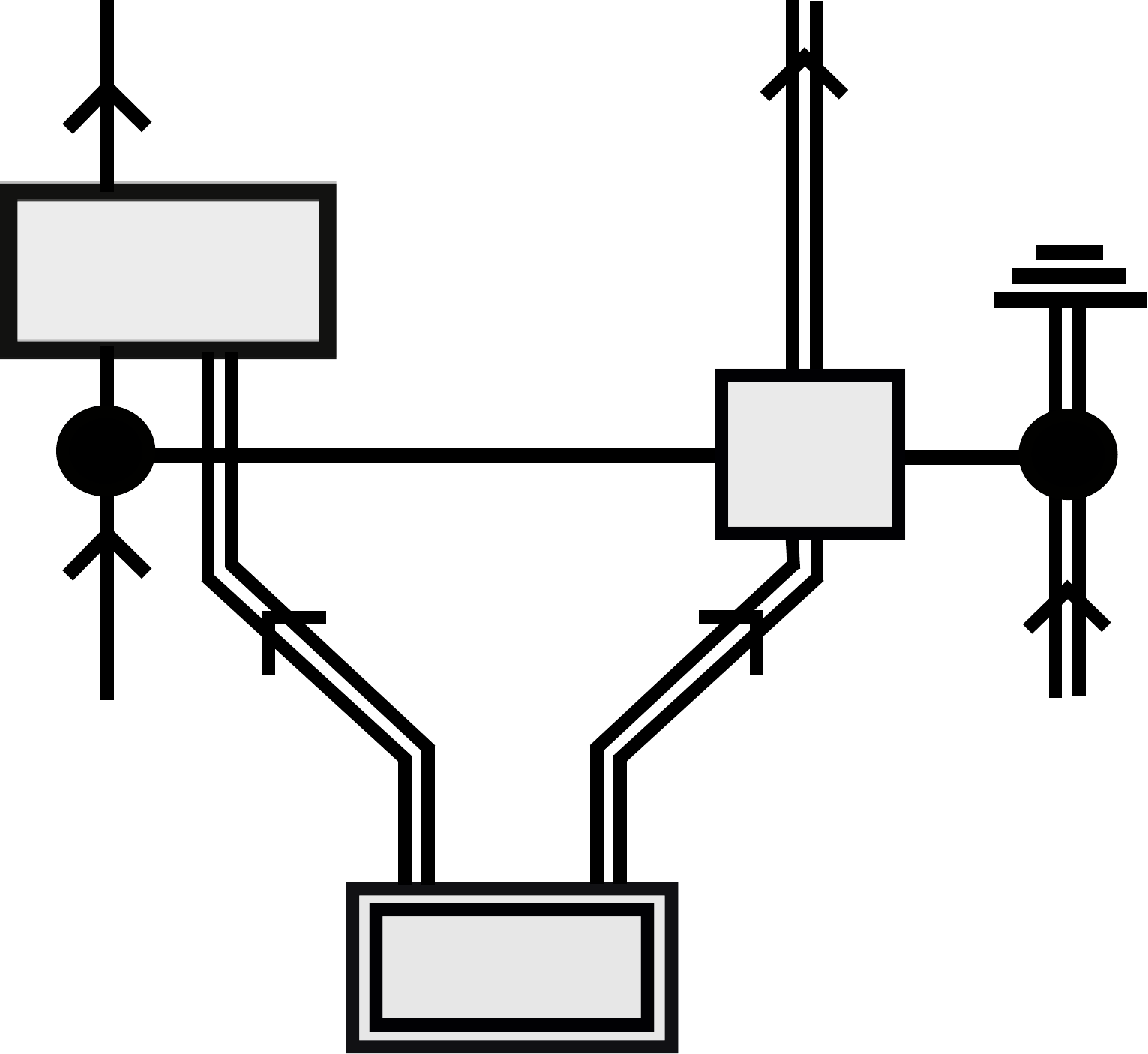}}
\put(-85,50){$x$}
\put(-85,115){$a$}
\put(11,67){X}
\put(55,50){$B_{in}$}
\put(-15,115){$B_{out}$}
\put(-26,7){$\rho_{AB}$}
\put(-26,-20){$(b)$}
}

  \end{center}
  \vspace*{-10mm}
  \caption{Mathematical depiction of two post-quantum channel assemblages. (a) Channel assemblage $\In^{PTP}$: Alice and Bob share a Bell state; Alice performs measurements on her system, while Bob performs a controlled-transpose operation. (b) Channel assemblage $\In^{PR}$: Alice and Bob share a Bell state; Alice performs measurements on her system, while Bob performs a CNOT operation controlled by both Alice's and Bob's inputs.  }
\end{figure}

The second example of a post-quantum channel assemblage that we consider may be mathematically expressed as: 
\begin{align}\label{eq:PR-channel}
\In^{PR}_{\A|\X} &= \left\{\I^{PR}_{a|x}\right\}_{a\in \A, \, x \in \X} \,,\\ 
\text{with} \quad &  \I_{a|x}^{PR}(\cdot) = \begin{cases} \tr{(M_{a} \otimes \id_{B_{out}}) \, \tr{(\id_A \otimes \text{CX}^{BB_{in} \rightarrow B_{out}}) \,(\rho_{AB} \otimes (\cdot)_{B_{in}})}{B}}{A} \,\text{if}\, x \in \{0,1\}, \nonumber \\
\I_{a|x}^{PR}(\cdot) = \frac{\id}{2}a \hskip21.5em\relax \text{if}\, x=2 \,, \nonumber \end{cases} \\
\text{and} \quad & M_{0} = \ket{0}\bra{0}\,,\quad M_{1} = \ket{1}\bra{1}\,, \nonumber 
\end{align}
where $\rho_{AB}$ is the same as in the previous example. This assemblage is illustrated in Fig.~\ref{fig:channel-PR}. If $x\in\{0,1\}$, Bob applies a controlled-X operation, where $x$ and $B_{in}$ are the control systems. More precisely, Bob applies a quantum instrument to his input system $B_{in}$ -- he measures $B_{in} $ in the computational basis, registers the classical output, which we here denote by $y$, and a quantum system containing the post-measurement state. Then, Bob transforms the system $B$ depending on the values of $x$ and $y$. If $xy=0$, he applies the identity map. If $xy=1$, Bob flips the system $B$. Finally, if $x=2$, Bob prepares the system $\frac{\id}{2}a$. For simplicity, hereon we denote $\In^{PR}_{\A|\X}=\In^{PR}$.

Note that Eq.~\eqref{eq:PR-channel} is merely meant as a mathematical description of the channel assemblage $\In^{PR}$, and is not meant to be taken as its experimental implementation. In particular, notice that the channel assemblage $\In^{PR}$ is no-signalling between Alice and Bob, contrary to what our chosen mathematical description may suggest. 

One may now argue that $\In^{PR}$ is a post-quantum assemblage. To to this, imagine that $B_{in}$ contains just a set of classical labels $\{\ket{0},\ket{1}\}$, and Alice and Bob generate $\In^{PR}$. Then, notice that for $x \in \{0,1\}$, if Bob performs a measurement on his subsystem on the $\{\ket{0},\ket{1}\}$ basis and observes a classical outcome $b$, Alice and Bob obtain correlations $p(ab|xy)$ that correspond to Popescu-Rohrlich (PR) box correlations \cite{popescu1994quantum}, which are known to be post-quantum. This shows that the assemblage $\In^{PR}$ is a  post-quantum channel assemblage. 

To study the relative order of $\In^{PTP}$ and $\In^{PR}$, we convert the assemblages into Choi form and run the SDP~\ref{SDP-channel} (in Matlab \cite{MATLAB:2010}, using the software CVX \cite{grant2013cvx,blondel2008recent}, the solver SeDuMi \cite{sturm1999using} and the toolbox QETLAB \cite{qetlab}; see the code at \cite{github}). We find that
the two assemblages are incomparable, which is summarized by the following observation:

\begin{obs}\label{obs:channel}
The two post-quantum channel assemblages $\In^{PR}$ and $\In^{PTP}$ are unordered resources in the LOSR resource theory of common-cause assemblages. 
\end{obs}

\section{The Bob-with-input EPR scenario}\label{sec:BWI}

The Bob-with-input EPR scenario, first introduced in Ref.~\cite{sainz2020bipartite}, is a special case of the channel EPR scenario. In this setting, Bob can locally influence the state preparation
of his system. This scenario is illustrated in Fig.~\ref{fig:BWI-scenario}. On the one hand, Alice acts on her share of the system by performing measurements and registering the obtained outcome -- this is identical to the role she plays in the channel EPR scenario. On the other hand, Bob chooses the value of a classical variable $y$, referred to as `Bob's input', which influences the state preparation of a quantum system in his laboratory. Operationally, one may think of Bob as holding a device that accepts the classical input $y$ (together with a physical system), and produces a quantum system prepared on some specified state. Bob's device's inner-workings can be thought of as a transformation of his subsystem (not necessarily a quantum one) into a new quantum system, where the transformation depends on the value of $y$. Notice that the system shared by Alice and Bob is depicted with a thick line in Fig.~\ref{fig:BWI-scenario} -- indeed, as mentioned regarding the channel EPR scenario, the GHJW theorem does not apply when Bob has an input; hence, one can find instances of \textit{post-quantum} Bob-with-input assemblages that do not admit a quantum realization \cite{sainz2015postquantum}. 

The relevant \textit{Bob-with-input assemblage} is now given by $\As_{\A|\X\Y}=\{\sigma_{a|xy}\}_{a,x,y}$, with $\tr{\sigma_{a|xy}}{}=p(a|x)$ and $\tr{\sum_a \sigma_{a|xy}}{}=1$ for all $x\in\X$ and $y\in\Y$. Notice that if $y$ takes only one value, this scenario coincides with the traditional EPR scenario for which an LOSR-based resource theory was developed in Ref.~\cite{EPRLOSR}.

\begin{figure}[h!]
  \begin{center}
  \subcaptionbox{\label{fig:BWI-scenario}}
{\put(-60,40){\includegraphics[width=0.22\textwidth]{Figs/BWI.pdf}}
\put(-45,80){$x$}
\put(38,80){$y$}
\put(-45,125){$a$}
\put(38,125){$\rho_{a|xy}$}
\put(-20,0){$(a)$}
}
\hspace{80mm} 
\subcaptionbox{\label{fig:BWI-scenario-LOSR}}
{\put(-85,0){\includegraphics[width=0.3\textwidth]{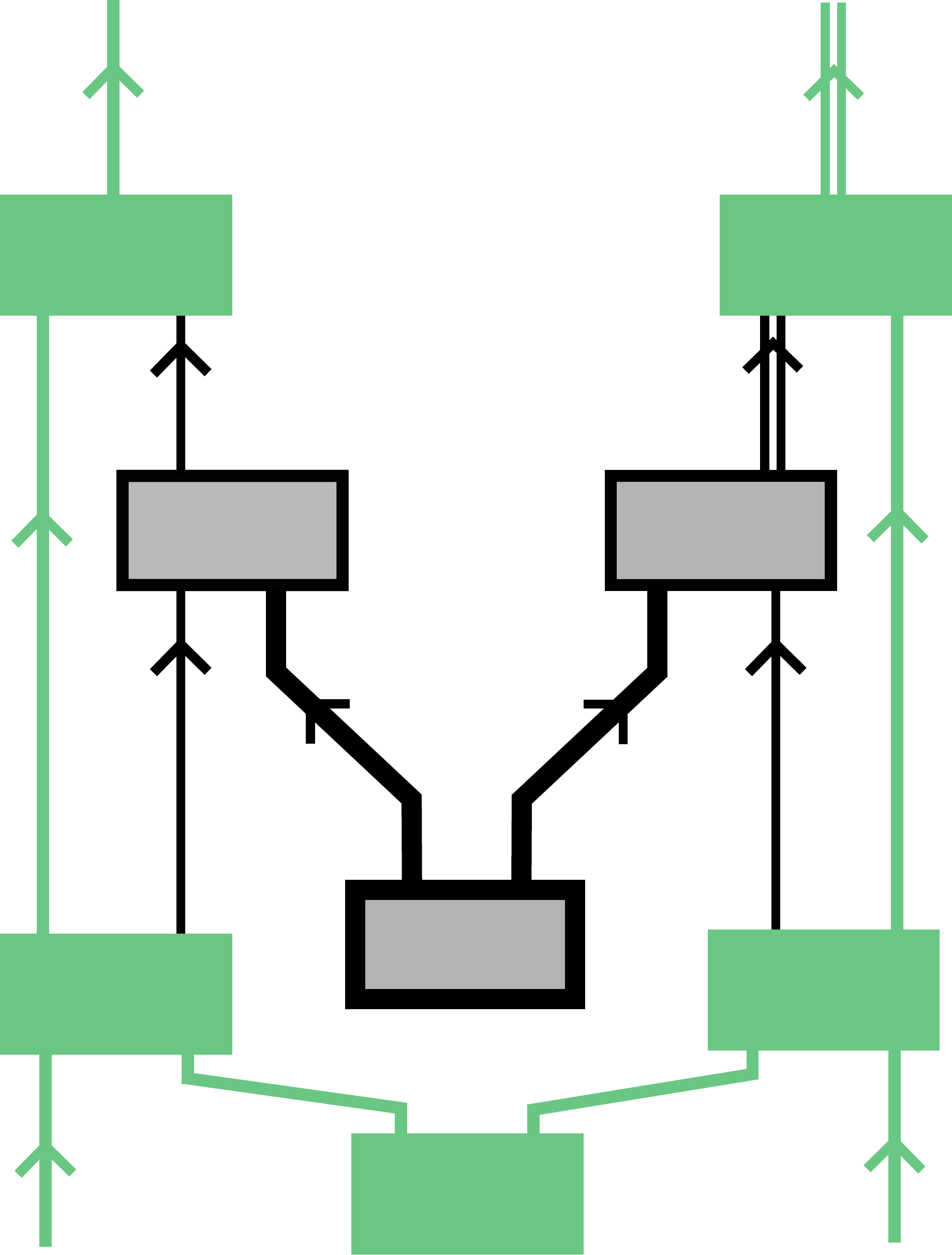}}
\put(-55,75){$x$}
\put(32,75){$y$}
\put(-55,120){$a$}
\put(32,120){$\rho_{a|xy}$}
\put(-72,10){$x^{\prime}$}
\put(50,10){$y^{\prime}$}
\put(-60,165){$a^{\prime}$}
\put(42,165){$\rho'_{a^{\prime}|x^{\prime}y^{\prime}}$}
\put(-25,-20){$(b)$}
}

  \end{center}
  \vspace*{-10mm}
  \caption{Depiction of a Bob-with-input EPR scenario. Systems that may be classical, quantum, or even post-quantum, are represented by thick lines. Quantum systems are represented by double lines and classical systems are depicted as single lines. (a) Bob-with-input assemblage: Alice and Bob share a (possibly post-quantum) common cause; Alice's input and output systems are classical while Bob's input is classical and output is quantum. (b) The most general LOSR operation on an assemblage in a Bob-with-input EPR scenario.  }
\end{figure}

If the common-cause Alice and Bob share is a quantum state, which we denote $\rho_{AB}$, the most general local operation Bob's device can implement is a collection of CPTP maps $\{ \xi_y \}_{y\in\Y}$. When Alice implements a POVM from the set $\{M_{a|x}\}_{a\in \A, x \in \X}$, we say that the elements of the assemblage admit a \textit{quantum realization} of the form \begin{align}\label{eq:quantum-BWI}
\sigma_{a|xy} =\xi_y[ \mathrm{tr}_A\{(M_{a|x} \otimes \id_B) \rho_{AB}\}]
\end{align}
for all $a \in \A,x \in \X$, and $y\in\Y$.

Similarly to the channel EPR scenario, nonclassical Bob-with-input assemblages are those that require a nonclassical common-cause. The classical Bob-with-input assemblages, i.e., the LOSR-free ones, can always be viewed as generated by local operations applied by each party that depend on the value of a shared classical random variable. Formally, a free assemblage in this scenario can be expressed as $\sigma_{a|xy}=\sum_{\lambda} p(\lambda) p(a|x\lambda)  \rho_{\lambda y}$. Here, $\lambda$ is the shared classical variable that is sampled according to $p(\lambda)$, $p(a|x\lambda)$ is a well-defined conditional probability distribution for all values of $\lambda$, and the quantum states $\rho_{\lambda y}$ are locally generated by Bob depending on the values of the classical variables $y$ and $\lambda$. This set of classical assemblages was first defined in Ref.~\cite{sainz2020bipartite}, where it is referred to as the set of `unsteerable' assemblages. Determining whether a given Bob-with-input assemblage is LOSR-free is possible with a single instance of an SDP, which is analogous both to SDP~\ref{SDP-channel-free} and (when $|\Y|=1$) to the SDP for checking `steerability' of a standard EPR assemblage given in Ref.~\cite{cavalcanti2016quantum}.

\subsection{LOSR transformations between Bob-with-input assemblages}

The most general LOSR transformation of a Bob-with-input assemblage is illustrated in Fig.~\ref{fig:BWI-scenario-LOSR}, where a local processing known as \textit{comb} (which locally pre- and post-processes the relevant systems in each wing) \cite{chiribella2009theoretical} with appropriate input/output system types is applied to each party. Notice how these are a special case of the processes of Fig.~\ref{fig:channel-LOSR} for the specific type of Bob's input system. 
This set of operations transforms one assemblage $\As_{\A|\X\Y}$ into a new assemblage $\As'_{\A'|\X'\Y'}$ as follows:
\begin{align}\label{eq:LOSRtrans-BWI} 
\begin{split}
\sigma'_{a'|x'y'}=\sum_{\lambda}  \sum_{a,x,y} 
p(a^{\prime},x|a,x^{\prime},\lambda) p(y|y^{\prime},\lambda) p(\lambda) \xi_{\lambda, y'}(\sigma_{a|xy}).
\end{split}
\end{align}
Here, $p(a^{\prime},x|a,x^{\prime},\lambda)$ encodes Alice's variables pre- and post-processing (this process is the same as in the channel EPR scenario), $p(y|y^{\prime},\lambda)$ encodes a classical pre-processing of Bob's classical input $y$ as a function of $y'$ and $\lambda$, and $\xi_{\lambda, y'}(\cdot)$ is the map corresponding to Bob’s local post-processing of his quantum system
as a function of $\lambda$ and $y'$.
Notice that, just like in the case of channel EPR scenarios, $p(a^{\prime},x|a,x^{\prime},\lambda)$ satisfies the no-retrocausation condition.

A simplified characterisation of a generic Bob-with-input LOSR transformation in terms of deterministic probability distributions $D(\cdot)$ is given by 
\begin{align}\label{eq:LOSRtrans-BWI2}
\sigma'_{a'|x'y'}=\sum_{\lambda}  \sum_{a,x,y} D(x|x^{\prime},\lambda) 
D(a^{\prime}|a,x^{\prime},\lambda) D(y|y^{\prime},\lambda)  \widetilde{\xi}_{\lambda,y^{\prime}}(\sigma_{a|xy})\,,
\end{align}
where $\widetilde{\xi}_{\lambda,y^{\prime}}(\cdot) = p(\lambda) \xi_{\lambda, y'}(\cdot)$.  
We will use this representation throughout this section; its detailed derivation is given in Appendix~\ref{app:deterministic-BwI}. In the Bob-with-Input scenario,  the total number of the deterministic strategies encoded in $\lambda$ is equal to $|\A^{\prime}|^{|\A|\times|\X'|}\times|\X|^{|\X'|}\times|\Y|^{|\Y'|}$. 

Given two assemblages, $\As_{\A|\X\Y}$ and $\As'_{\A^{\prime}|\X^{\prime}\Y^{\prime}}$, 
deciding whether $\As_{\A|\X\Y}$ can be converted into $\As'_{\A^{\prime}|\X^{\prime}\Y^{\prime}}$ under LOSR operations is equivalent to checking whether the elements of $\As'_{\A^{\prime}|\X^{\prime}\Y^{\prime}}$ admit a decomposition as per Eq.~\eqref{eq:LOSRtrans-BWI2}. The Bob-with-input EPR scenario is simply a special case of the channel EPR scenario wherein Bob's input is a classical system. However, as the Bob-with-input EPR scenario exhibits a simpler characterization of an LOSR transformation of an assemblage, the SDP for testing resource conversion in this scenario can be simplified compared to SDP~\ref{SDP-channel}, as we show below.

Notice that the CPTNI map $\widetilde{\xi}_{\lambda,y^{\prime}}(\sigma_{a|xy})$ can be represented in terms of its (possibly subnormalized) Choi state $J_{\xi\,\lambda\,y'}$ as follows: 
\begin{align}
\widetilde{\xi}_{\lambda,y^{\prime}}(\sigma_{a|xy}) = d_B\, \tr{J_{\xi\,\lambda\,y'}\,(\id_{B^{\prime}} \otimes \sigma_{a|xy}^{T})}{B}\,,
\end{align}
where Bob's output system is defined on $\cH_{B'}$. Therefore, for $\As'_{\A'|\X'\Y'}$ to admit a decomposition as per Eq.~\eqref{eq:LOSRtrans-BWI2}, each $\sigma'_{a'|x'y'}$ must decompose as 
\begin{align}\label{eq:sigSDP-BWI}
    \sigma'_{a^{\prime}|x^{\prime}y^{\prime}}= \sum_{\lambda}\sum_{a,x,y}  D(x|x^{\prime},\lambda) 
D(a^{\prime}|a,x^{\prime},\lambda) D(y|y^{\prime},\lambda) 
\,d_B\, \tr{J_{\xi\,\lambda\,y'}\,(\id_{B^{\prime}} \otimes \sigma_{a|xy}^{T})}{B}\,.
\end{align}
We implement this condition in the following SDP:

\begin{sdp}\textbf{$\As_{\A|\X\Y} \overset{\text{LOSR}}{\longrightarrow} \, \As'_{\A'|\X'\Y'}$.}\label{sdp:BWI}\\
The assemblage $\As_{\A|\X\Y}$ can be converted into the assemblage $\As'_{\A'|\X'\Y'}$ under LOSR operations if and only if the following SDP is feasible:
\begin{align}
\begin{split}
\textrm{given} \;\;\;& \{ \sigma_{a|xy}\}_{a,x,y}\,,\; \{ \sigma'_{a^{\prime}|x^{\prime}y^{\prime}}\}_{a^{\prime},x^{\prime},y^{\prime}}\,,\; \{D(x|x^{\prime},\lambda) \}_{\lambda,x,x^{\prime}} \,,\; \{D(a^{\prime}|a,x^{\prime},\lambda)\}_{\lambda,a^{\prime},a,x^{\prime}} \,,\; \{D(y|y^{\prime},\lambda)\}_{\lambda,y,y^{\prime}} \\
    \textrm{find} \;\;\;& \{(J_{\xi\,\lambda\,y'})_{BB'}\}_{\lambda,y^{\prime}}  \\
    \textrm{s.t.} \;\;\;& \begin{cases} J_{\xi\,\lambda\,y'} \geq 0 \quad \forall \, \lambda,y^{\prime}\,,\\  
      \tr{J_{\xi\,\lambda\,y'}}{B^{\prime}} \propto \id_{B} \quad \forall \, \lambda,y^{\prime}\,, \\
      \sum_{\lambda} \tr{J_{\xi\,\lambda\,y'}}{B^{\prime}} = \frac{1}{d}\, \id_{B}\quad \forall \, y' \,,\\
      \tr{J_{\xi\,\lambda\,y'_1}}{B^{\prime}} = \tr{J_{\xi\,\lambda\,y'_2}}{B^{\prime}} \quad \forall \, \lambda\,, y^{\prime}_1, y^{\prime}_2\,, \\
      \sigma'_{a^{\prime}|x^{\prime} y^{\prime}}= \sum_{\lambda}\sum_{a,x,y} D(x|x^{\prime},\lambda)\, D(a^{\prime}|a,x^{\prime},\lambda)\, D(y|y^{\prime},\lambda) \, d_B\, \tr{J_{\xi\,\lambda\,y'}\,(\id_{B^{\prime}} \otimes \sigma_{a|xy}^{T})}{B}\,.
     \end{cases}
    \end{split}
\end{align}

When the conversion is not possible, we denote it by $\As_{\A|\X\Y} \, \overset{\text{LOSR}}{\not\longrightarrow} \, \As'_{\A^{\prime}|\X^{\prime}\Y^{\prime}}$.
\end{sdp}
For the robust formulation of this SDP, see Appendix~\ref{app:SDP}.

\subsection{Properties of the pre-order}\label{sec:preorder-BwI}

In analogy to the channel EPR scenario, studying post-quantum Bob-with-input assemblages gives us insight into the pre-order of resources. In this section, we introduce four Bob-with-input assemblages and study the possible conversions between them, both analytically and using SDP~\ref{sdp:BWI}. 

In Section~\ref{sec:preorder-channel}, we focused on accessing the pre-order of channel assemblages using SDP~\ref{SDP-channel}. Due to the simpler nature of the Bob-with-input EPR scenario compared to the channel EPR scenario, here we first focus on studying the pre-order analytically. We start with a Bob-with-input EPR scenario where $\X=\{0,1,2\}$, $\A=\{0,1\}$ and $\Y=\{0,1\}$, and we consider conversions between two post-quantum Bob-with-input assemblages introduced in Ref. \cite{sainz2020bipartite}. These two assemblages are special cases of channel assemblages $\In^{PTP}$ and $\In^{PR}$, where Bob's input states are just elements of an orthonormal basis $\{\ket{0},\ket{1}\}$. 

For our first post-quantum Bob-with-input assemblage, consider the Bob-with-input assemblage studied in Ref.~\cite[Eq.~(6)]{sainz2020bipartite}, which is a special case of $\In^{PTP}$. This assemblage can be mathematically expressed as follows: 
\begin{align}\label{eq:PTP}
\As^{PTP}_{\A|\X\Y} &= \left\{\sigma^{PTP}_{a|xy}\right\}_{a\in \A, \, x \in \X, \, y \in \Y} \,,\\
\text{with} \quad & \begin{cases} \sigma^{PTP}_{a|xy} = \xi_y \, \{ \tr{({M}_{a|x} \otimes \id_B)  \, \ket{\phi}\bra{\phi}}{\mathrm{A}}\} \,, \nonumber \\
{M}_{a|1} = \frac{\id + (-1)^a \sigma_x}{2}\,,\quad {M}_{a|2} = \frac{\id + (-1)^a \sigma_y}{2}\,,\quad
{M}_{a|3} = \frac{\id + (-1)^a \sigma_z}{2}\,, \nonumber \end{cases}
\end{align}
where $\ket{\phi} = \frac{\ket{00}+\ket{11}}{\sqrt{2}}$ and $\sigma_x$, $\sigma_y$, and $\sigma_z$ are the Pauli operators. The map $\xi_y$ is the following: the identity quantum channel for $y=0$, and the transpose operation for $y=1$. For simplicity, in this section we denote $\As^{PTP}_{\A|\X\Y}=\As^{PTP}$. In Ref.~\cite[Appendix D]{sainz2020bipartite}, it was shown that $\As^{PTP}$ is a post-quantum assemblage. There it was moreover shown that, if Bob decides to measure his subsystem, the correlations that arise between him and Alice always admit a quantum explanation. This will prove relevant for the resource-conversion statements in this manuscript. 

As a second example of a post-quantum Bob-with-input assemblage, we consider a special case of $\In^{PR}$. We follow the construction introduced in Ref.~\cite[Eq. (5)]{sainz2020bipartite}, and define:
\begin{align}\label{eq:PR}
\As^{PR}_{\A|\X\Y} &= \left\{\sigma^{PR}_{a|xy}\right\}_{a\in \A, \, x \in \X, \, y \in \Y} \,,\\
\text{with} \quad \sigma^{PR}_{a|xy} & = \left\{ \begin{array}{ll}
\ket{a \oplus xy}\bra{a \oplus xy} & \textrm{if $x \in \{0,1\}$}\\
\frac{\id}{2} a & \textrm{if $x=2.$}
\end{array} \right.
\nonumber
\end{align}
For simplicity, hereon we denote $\As^{PR}_{\A|\X\Y}=\As^{PR}$. As we already pointed out in the channel EPR scenario, for $x \in \{0,1\}$, if Bob decides to measure his subsystem in the computational basis and registers a classical outcome $b$, Alice and Bob obtain PR box correlations, what certifies that the assemblage $\As^{PR}$ is post-quantum. 

We will now show the the two post-quantum assemblages $\As^{PR}$ and $\As^{PTP}$ are unordered in our LOSR resource theory. 

\begin{thm}\label{thm:PTP-PR}
$\As^{PTP}$ cannot be converted into $\As^{PR}$ with LOSR operations. 
\end{thm}
\begin{proof}
Let us prove this by contradiction. Assume that an LOSR-processing of $\As^{PTP}$ yields $\As^{PR}$. Then, since $\As^{PR}$ can generate post-quantum correlations in a Bell-type experiment, it follows that an LOSR-processing of $\As^{PTP}$ can also generate post-quantum correlations. However, we know that $\As^{PTP}$ can only generate quantum correlations \cite[Appendix D]{sainz2020bipartite}. As LOSR operations cannot create post-quantum Bell non-locality, this contradicts the initial assumption and hence proves the claim. 
\end{proof}

\begin{thm}\label{thm:PR-PTP}
$\As^{PR}$ cannot be converted into $\As^{PTP}$ with LOSR operations. 
\end{thm}

The proof of this theorem is given in Appendix~\ref{app:proof-thm}. In the proof we use the `steering' functional constructed in Ref.~\cite[Eq.~(D3)]{sainz2020bipartite} which achieves its minimum value when evaluated on $\As^{PTP}$. We show that neither $\As^{PR}$ or any LOSR-processing of $\As^{PR}$ can achieve the minimum value of this `steering' functional, which completes the proof.

\begin{cor}\label{cor:BWI}
The two post-quantum assemblages $\As^{PR}$ and $\As^{PTP}$ are unordered resources in the LOSR resource theory of common-cause assemblages. 
\end{cor}

This result can be verified with SDP~\ref{sdp:BWI} (we verified it in Matlab \cite{MATLAB:2010}, using the software CVX \cite{grant2013cvx,blondel2008recent}, the solver SeDuMi \cite{sturm1999using} and the toolbox QETLAB \cite{qetlab}; see the code at \cite{github}). 

It is usually the case that the pre-order in a given resource theory is studied using \textit{resource monotones}. Interestingly, Corollary~\ref{cor:BWI} does not rely on a construction of resource monotones, but it arises from specific considerations of LOSR transformations.

We now move on to introducing two more examples of Bob-with-input assemblages. We start by introducing an assemblage that, as far as we are aware, has not been studied in the literature before. For this purpose, let us briefly recall the meaning of almost-quantum correlations. 

The set of almost-quantum correlations $\widetilde{Q}$~\cite{AQ} was originally proposed as a set of correlations that satisfies numerous principles of a reasonable physical theory, including information causality~\cite{pawlowski2009information}, macroscopic locality~\cite{navascues2010glance} and local orthogonality~\cite{sainz2014exploring}. It was first defined for Bell-type scenarios and it was showed to be larger than the set of quantum correlations and to strictly contain them. The concept of almost-quantum correlations was later generalized to other physical set-ups~\cite{acin2015combinatorial,henson2015macroscopic},  including EPR scenarios~\cite{sainz2015postquantum,sainz2020bipartite}. Within this generalization, a particular relaxation of the definition of quantum assemblages allows one to construct almost-quantum assemblages.

From now on let us focus on a Bob-with-input scenario where $\X=\{0,1\}$, $\A=\{0,1\}$ and $\Y=\{0,1\}$. To construct an example of a post-quantum assemblage, consider the probability distribution generated in a bipartite Bell scenario introduced in Ref.~\cite{AQ}, which we denote $\vec{p}_{AQ}$. We recall the exact form of $\vec{p}_{AQ}$ in Appendix~\ref{app:AQ}; for now, the important property of $\vec{p}_{AQ}$ to note is that it is post-quantum and it lives in $\widetilde{Q}$. Consider the following Bob-with-input assemblage:

\begin{align}\label{eq:AQ}
\As^{AQ}_{\A|\X\Y} &= \left\{\sigma^{AQ}_{a|xy}\right\}_{a\in \A, \, x \in \X, \, y \in \Y} \,,\\
\text{with} \quad   \sigma^{AQ}_{a|xy} &= \sum_b \, p_{AQ}(ab|xy) \ket{b} \bra{b}  \,, \nonumber 
\end{align}
where the elements $p_{AQ}(ab|xy)$ can be read from the vector $\vec{p}_{AQ}$. For simplicity, we denote $\As^{AQ}_{\A|\X\Y}=\As^{AQ}$. This assemblage is clearly not quantum-realizable. Indeed, if Bob chooses to measure his system in the computational basis, i.e., $N_b = \ket{b}\bra{b}$, the correlations that Alice and Bob obtain are given by $p(ab|xy)=\tr{N_b \sigma_{a|xy}}{}$, which gives exactly $p_{AQ}(ab|xy)$. This shows that $\As^{AQ}$ is post-quantum\footnote{It is important to note that the assemblage $\As^{AQ}$ is not necessarily an almost-quantum assemblage. The set of almost-quantum assemblages is a strict subset of the set of post-quantum assemblages. If Bob decides to measure his subsystem with a measurement different than $N_b = \ket{b}\bra{b}$, it might be possible that post-quantum correlations that are not almost-quantum are generated.}, since for any quantum-realizable assemblage, Alice and Bob can only generate quantum correlations if Bob decides to measure his system.

As the next example, consider the assemblage $\As^{\prime \, PR}$ built from $\As^{PR}$ by considering the assemblage elements of the latter that correspond to $x=0,1$. 
To study the relative order of $\As^{AQ}$ and $\As^{\prime \, PR}$, we run the SDP~\ref{sdp:BWI} (see the code at \cite{github}). We find that the two assemblages are strictly ordered, which is summarized by the following observation:

\begin{obs}\label{obs:BWI-AQ-PR}
The post-quantum Bob-with-input assemblage $\As^{\prime \, PR}$ is strictly above $\As^{AQ}$ in the pre-order of resources in the LOSR resource theory of common-cause assemblages. 
\end{obs}
In Appendix~\ref{app:proofs-obs}, we analytically show that $ \As^{\prime \, PR} \, \overset{\text{LOSR}}{\longrightarrow} \,\As^{AQ}$. 

For our final example in a Bob-with-input scenario, we consider the following quantum-realizable assemblage:

\begin{align}\label{eq:CHSH}
\As^{CHSH}_{\A|\X\Y} &= \left\{\sigma^{CHSH}_{a|xy}\right\}_{a\in \A, \, x \in \X, \, y \in \Y} \,,\\
\text{with} \quad & \begin{cases} \sigma^{CHSH}_{a|xy} =  \tr{(M_{a|x} \otimes \id_B)  \,
 \ket{\phi}\bra{\phi}}{\mathrm{A}} \,, \nonumber  \\
\ket{\phi} = \frac{\ket{00}+\ket{11}}{\sqrt{2}} \,, \nonumber \\
M_{a|0} = \frac{\id + (-1)^a \sigma_z}{2}\,,\quad M_{a|1} = \frac{\id + (-1)^a \sigma_x}{2}\,, \nonumber  \end{cases}
\end{align}
where  $\sigma_x$, $\sigma_y$, and $\sigma_z$ are the Pauli operators. 
For simplicity, we denote $\As^{CHSH}_{\A|\X\Y}=\As^{CHSH}$. Notice that if Alice and Bob generate $\As^{CHSH}$, and then Bob decides to measure his system with suitable measurements, they will obtain correlations that violate the $CHSH$ Bell inequality maximally \cite{clauser1969proposed}.

The relation between $\As^{CHSH}$ and $\As^{AQ}$, obtained using SDP~\ref{sdp:BWI}, is the following:

\begin{obs}\label{obs:BWI-AQ-CHSH}
The two Bob-with-input assemblages $\As^{CHSH}$ and $\As^{AQ}$ are unordered resources in the LOSR resource theory of common-cause assemblages. 
\end{obs}

Notice that the direction $\As^{CHSH} \, \overset{\text{LOSR}}{\not\longrightarrow} \,\As^{AQ}$ is straight-forward to prove by noting that LOSR operations cannot create post-quantumness. 

\section{The measurement-device-independent EPR scenario}\label{sec:MDI}

We will now consider a special case of a channel EPR scenario in which Bob has a measurement channel rather than a general quantum channel. This measurement-device-independent (MDI) EPR scenario is illustrated in Fig.~\ref{fig:MDI}. Alice and Bob share a system $AB$. Alice's role is still the same as in the channel and Bob-with-input scenarios: she performs measurements $\{M_{a|x}\}_{a\in \A, x \in \X}$ to obtain a classical output $a$. Now, Bob holds a collection of measurement channels\footnote{ Here, by \textit{measurement channel} is meant a \textit{quantum instrument} with a trivial output Hilbert space.}, which we denote by $\{\Omega_b^{B_{in} \rightarrow B_{out}}\}_{b\in \B}$, where the output system $B_{out}$ is just a classical variable that may take values within $\B$. In a way one can think of these measurement channels as a measuring device that implements a single generalised measurement (as given by a POVM) and keeps a record of the obtained outcome. 

Similarly to the channel EPR scenario, we are here interested in the case where the measurement channels that Bob has access to may in addition be correlated with some physical system in Alice's lab. This situation is formalised by the premise that Bob has instead access to a processing $\Theta_b^{BB_{in} \rightarrow B_{out}}$ that takes as input system his own $B_{in}$ together with the half of the system $AB$ he shares with Alice (denoted by $B$). The marginal measurement channel  $\{\Omega_b^{B_{in} \rightarrow B_{out}}\}_{b\in \B}$ that was introduced at the beginning is therefore a function of these superseding measurement apparatus $\Theta_b^{BB_{in} \rightarrow B_{out}}$ and the state of the system $AB$ shared by Alice and Bob (the specific dependence is specified further below). The idea is then to see what Alice can infer about the measurement channel $\Theta_b^{BB_{in} \rightarrow B_{out}}$ from the local measurements she is performing on her share of $AB$.

\begin{figure}[h!]
  \begin{center}
  \subcaptionbox{\label{fig:MDI}}
{\put(-60,40){\includegraphics[width=0.22\textwidth]{Figs/MDI.pdf}}
\put(-46,80){$x$}
\put(-46,125){$a$}
\put(38,125){$b$}
\put(-20,150){$\Ne_{ab|x}(\cdot)$}
\put(38,80){$\cH_{B_{in}}$}
\put(-20,0){$(a)$}
}
\hspace{80mm} 
\subcaptionbox{\label{fig:MDI-LOSR}}
{\put(-85,0){\includegraphics[width=0.3\textwidth]{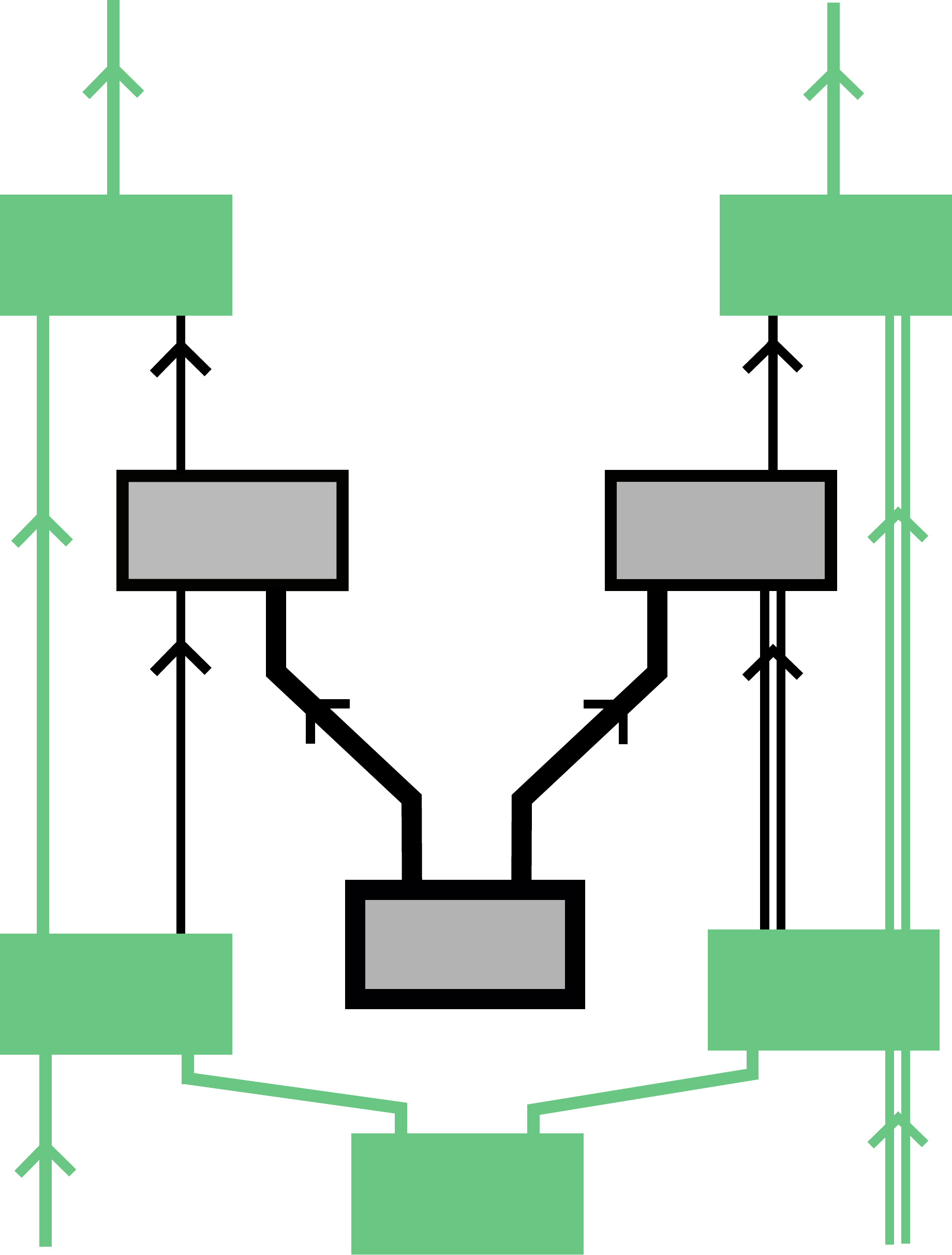}}
\put(-55,75){$x$}
\put(-55,120){$a$}
\put(12,120){$b$}
\put(-72,10){$x^{\prime}$}
\put(-60,165){$a^{\prime}$}
\put(-25,180){$\Ne'_{a^{\prime}b^{\prime}|x^{\prime}}(\cdot)$}
\put(42,165){$b'$}
\put(0,50){$\cH_{B_{in}}$}
\put(47,10){$\cH_{B_{in'}'}$}
\put(-25,-20){$(b)$}
}

  \end{center}
  \vspace*{-10mm}
  \caption{Depiction of a measurement-device-independent EPR scenario. Systems that may be classical, quantum, or even post-quantum, are represented by thick lines.  Quantum systems are represented by double lines and classical systems are depicted as single lines. (a) MDI assemblage: Alice and Bob share a (possibly post-quantum) common cause; Alice's input and output systems are both classical, while Bob's input and output systems are quantum and classical, respectively. (b) The most general LOSR operation on an assemblage in a MDI EPR scenario. }
\end{figure}

Formally, the relevant \textit{MDI assemblages} in this scenario are given by $\N_{\A\B|\X}=~\{\Ne_{ab|x}(\cdot)\}$ with $a\in\A\,,b\in\B$ and $x\in\X$, where $\Ne_{ab|x}(\cdot)$ is a channel -- with trivial output space -- associated to the POVM element corresponding to outcome $b$ of Bob's measurement device (when Alice's measurement event is $a|x$). The elements of a valid MDI assemblage satisfy the conditions $\sum_b \Ne_{ab|x}(\cdot) = p(a|x)$ and $\sum_a \Ne_{ab|x}(\cdot) = \Omega_b^{B_{in} \rightarrow B_{out}}$. These conditions, equivalent to no-signalling requirements between Alice and Bob, define hence the enveloping theory of resources.  A natural question then is when a given MDI assemblage has a classical, quantum, or post-quantum realisation, as we briefly recall next. 

Let us begin with quantumly-realisable MDI assemblages, that is, those compatible with Alice and Bob sharing a quantum common cause. When Alice and Bob share a bipartite quantum system prepared on state $\rho_{AB}$, the elements of the MDI assemblage are given by
\begin{equation}\label{eq:MDI-ass}
\Ne_{ab|x}(\cdot)=\tr{(M_{a|x} \otimes \id_{B_{out}}) (\id_A \otimes \Theta_b^{BB_{in} \rightarrow B_{out}})[\rho_{AB} \otimes (\cdot)]}{}.
\end{equation}
Notably,  MDI assemblages are not always compatible with a quantum common cause, i.e., there exist post-quantum MDI assemblages\footnote{ Although Ref.~\cite{hoban2018channel} does not discuss measurement-device-independent assemblages directly, this scenario naturally appears if one wants to compare Buscemi non-locality and EPR scenarios.} \cite{hoban2018channel}.

Now, let us specify the free sub-theory of resources: those that may arise by the parties implementing local operations and shared randomness.  Free MDI assemblages in our resource theory are those that can be generated with a classical common cause. They can be thought of as objects generated by Alice's local measurements and Bob's local measurement channel, where both parties are correlated by a classical variable. Formally, the elements of an LOSR-free MDI assemblage can be written as $\Ne_{ab|x}(\cdot)=\sum_{\lambda} p(\lambda) p(a|x\lambda) \Ne_{b,{\lambda}}(\cdot)$, where the measurement channels $\Ne_{b,{\lambda}}$ depend on the value of the classical common cause $\lambda$. 
  
\subsection{LOSR transformations between  measurement-device-independent assemblages}\label{sec:MDI-LOSR}

The most general LOSR transformation of an MDI assemblage is presented in Fig.~\ref{fig:MDI-LOSR}. On Alice's side, the LOSR map is exactly the same as in the channel and Bob-with-input scenarios. On Bob's side, however, the processing has appropriate input and output system types adjusted to the MDI scenario. Let us express Bob's processing as a single CPTNI map $\zeta_{b\,b'\,\lambda}^{B_{in'}' \rightarrow B_{in}}$. This map has two inputs (one classical, $b$, and one quantum, $B_{in'}'$) and two outputs (again, one is classical, $b'$, and one is quantum, $B_{in}$). Additionally, the map $\zeta_{b\,b'\,\lambda}^{B_{in'}' \rightarrow B_{in}}$ must satisfy the  no-retrocausation condition, i.e., it must be such that $b$ does not influence $B_{in}$. This condition means that $\sum_{\lambda,b'}\zeta_{b\,b'\,\lambda}^{B_{in'}' \rightarrow B_{in}} $ must be a CPTP map. 

Similarly to the channel and Bob-with-input scenarios, it is convenient to use the Choi-Jamiołkowski isomorphism here. Let $J_{a'b'|x'}$,  $J_{ab|x}$ and $J_{\zeta\,b\,b'\, \lambda}$ correspond to the Choi representation of maps $\Ne'_{a'b'|x'}$, $\Ne_{ab|x}$ and $\zeta_{b\,b'\,\lambda}$, respectively. Then, the most general LOSR processing of a MDI assemblage can be written as
\begin{align}\label{eq:LOSR-MDI} 
\begin{split}
J_{a'b'|x'}=\sum_{\lambda}  \sum_{a,x} 
D(a^{\prime}|a,x^{\prime},\lambda)\,D(x|x^{\prime},\lambda)\,  J_{ab|x}\, * J_{\zeta\,b\,b'\, \lambda} \,.
\end{split}
\end{align}
Here, $D(\cdot)$ are deterministic probability distributions (see previous sections and Appendix~\ref{app:deterministic-channel} for details) and the total number of the deterministic strategies encoded in $\lambda$ is equal to $|\A^{\prime}|^{|\A|\times|\X'|}\times|\X|^{|\X'|}$. As mentioned before, the link product is given by
\begin{align}\label{eq----} 
 J_{ab|x}\, * J_{\zeta\,b\,b'\, \lambda} = d_{B_{in}} \,  \,\tr{ ( \id_{B_{in'}'} \otimes J_{ab|x}) \, J_{\zeta\,b\,b'\, \lambda}^{T_{B_{in}}}}{B_{in}}\,.
\end{align}

Deciding if an assemblage $\N_{\A\B|\X}$ can be converted into a different assemblage $\N'_{\A'\B'|\X'}$ with LOSR operations comes down to checking  whether there exist a collection of Choi states $\{J_{\zeta\,b\,b'\,\lambda}\}_{b,b',\lambda}$ such that the Choi form of $\N'_{\A'\B'|\X'}$ can be decomposed as in Eq.~\eqref{eq:LOSR-MDI}. This can be decided with a single instance of the following semidefinite program:

\begin{sdp}\label{SDP-MDI}
The MDI assemblage $\N_{\A\B|\X}$ can be converted to the MDI assemblage $\N'_{\A'\B'|\X'}$ under LOSR operations, denoted by $\N_{\A\B|\X} \, \overset{\text{LOSR}}{\longrightarrow} \, \N'_{\A^{\prime}\B'|\X^{\prime}}$, if and only if the following SDP is feasible:
\begin{align}
\begin{split}
\textrm{given} \;\;\;& \{ J_{ab|x}\}_{a,b,x}\,,\; \{J'_{a'b'|x'}\}_{a',b',x'}\,,\; \{D(a^{\prime}|a,x^{\prime},\lambda)\}_{\lambda,a^{\prime},a,x^{\prime}} \,, \;  \{D(x|x^{\prime},\lambda)\}_{\lambda,x,x^{\prime}} \\
    \textrm{find} \;\;\;& \{J_{\zeta\,b\,b'\,\lambda}\}_{b\,,b'\,,\lambda} \\
    \textrm{s.t.} \;\;\;& \begin{cases} 
   J_{\zeta\,b\,b'\,\lambda} \geq 0 \quad \forall\,b,b',\lambda\,,\\
     \sum_{b'} \tr{J_{\zeta\,b\,b'\,\lambda}}{B_{in}} \propto \id_{B_{in'}'} \quad \forall\,b,\lambda \,, \\
      \sum_{\lambda,b'} \tr{J_{\zeta\,b\,b'\,\lambda}}{B_{in}} = \frac{1}{d_{B_{in'}'}}\, \id_{B_{in'}'} \quad \forall\,b\,, \\
      \sum_{b'} \, J_{\zeta\,b\,b'\, \lambda} \geq 0  \quad \forall\,b,\lambda\,,\\
            J_{a'b'|x'}=\sum_{\lambda}  \sum_{a,b,x} 
D(a^{\prime}|a,x^{\prime},\lambda)\,D(x|x^{\prime},\lambda)\,  J_{ab|x}\, * J_{\zeta\,b\,b'\, \lambda} \,.
      \end{cases}
    \end{split}
\end{align}
When the conversion is not possible, we denote it by $\N_{\A\B|\X} \, \overset{\text{LOSR}}{\not\longrightarrow} \, \N'_{\A^{\prime}\B'|\X^{\prime}}$.
\end{sdp}
For the discussion of the robustness of this SDP, see Appendix~\ref{app:SDP}.

\subsection{Properties of the pre-order}\label{sec:preorder-MDI} 

Similarly to the scenarios studied in the previous sections, we can use semidefinite programming to study possible conversions among post-quantum measurement-device-independent assemblages. In this section, we focus on an MDI scenario where $\X=\{0,1,2\}$, $\A=\{0,1\}$ and $\B=\{0,1\}$. We consider conversions between two post-quantum MDI assemblages that are special cases of channel assemblages $\In^{PTP}$ and $\In^{PR}$.

Let us start with a post-quantum MDI assemblage that can be mathematically expressed as if arising from a controlled-transpose operation. Such assemblage can be expressed as:
\begin{align}\label{eq:PTP-MDI}
\N^{PTP}_{\A\B|\X} &= \left\{\Ne^{PTP}_{ab|x}\right\}_{a\in \A, \, b \in \B, \, x \in \X} \,,\\
\text{with} \quad & \begin{cases} \Ne^{PTP}_{ab|x} = \tr{({M}_{a|x} \otimes \id_{B_{out}}) (\id_A \otimes \textrm{CT}_b^{BB_{in} \rightarrow B_{out}})[\rho_{AB} \otimes (\cdot)]}{} \,, \nonumber \\
{M}_{a|1} = \frac{\id + (-1)^a \sigma_x}{2}\,,\quad {M}_{a|2} = \frac{\id + (-1)^a \sigma_y}{2}\,,\quad
{M}_{a|3} = \frac{\id + (-1)^a \sigma_z}{2}\,, \nonumber \end{cases}
\end{align}
where $\rho_{AB}=\ket{\phi}\bra{\phi}$ with $\ket{\phi} = \frac{\ket{00}+\ket{11}}{\sqrt{2}}$ and $\sigma_x$, $\sigma_y$, and $\sigma_z$ are the Pauli operators. Here, the processing CT$_b^{BB_{in} \rightarrow B_{out}}$ is the following. First, a controlled transpose operation is applied on $BB_{in}$, where $B_{in}$ is the control qubit and $B$ is the system that is being transposed. Second, the system $B_{in}$ is traced-out and the system $B$ is measured by Bob. The measurement elements are $N_0=\frac{1}{3}\id+\frac{1}{3}\sigma_y$ and $N_1=\frac{2}{3}\id-\frac{1}{3}\sigma_y$. The outcome of Bob's measurement is defined on $B_{out}$. This process is illustrated in Fig.~\ref{fig:MDI-PTP}. For simplicity, we hereon denote $\N^{PTP}_{\A\B|\X}$ as $\N^{PTP}$. In Appendix~\ref{app:MDI}, we prove that $\N^{PTP}$ is postquantum: we construct an SDP that tests a membership of a measurement-device-independent assemblage to the relaxation of a quantum set. 

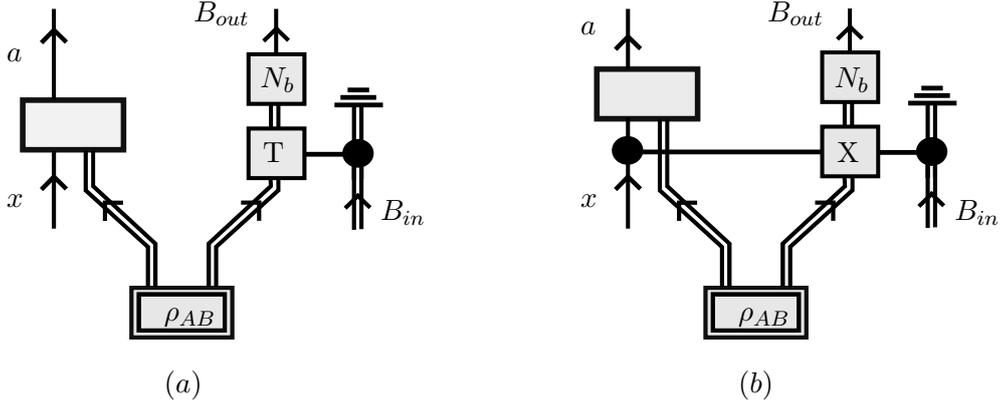
\begin{figure}[h!]
  \begin{center}
  \subcaptionbox{\label{fig:MDI-PTP}}
{\put(-60,20){\includegraphics[width=0.3\textwidth]{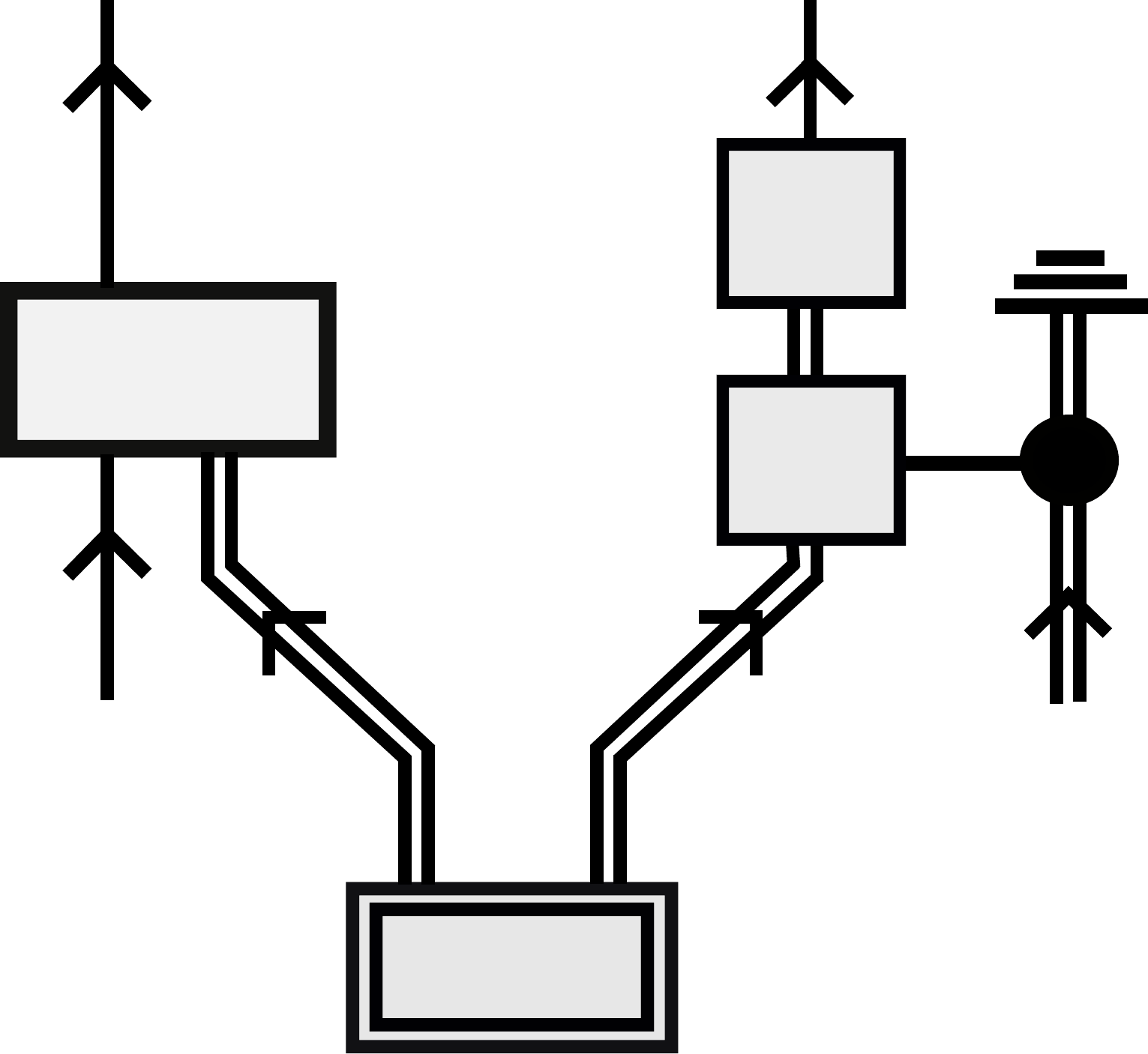}}
\put(-65,70){$x$}
\put(-65,125){$a$}
\put(5,140){$B_{out}$}
\put(75,65){$B_{in}$}
\put(30,115){$N_b$}
\put(31,87){T}
\put(-6,27){$\rho_{AB}$}
\put(-6,0){$(a)$}
}
\hspace{80mm} 
\subcaptionbox{\label{fig:MDI-PR}}
{\put(-80,0){\includegraphics[width=0.3\textwidth]{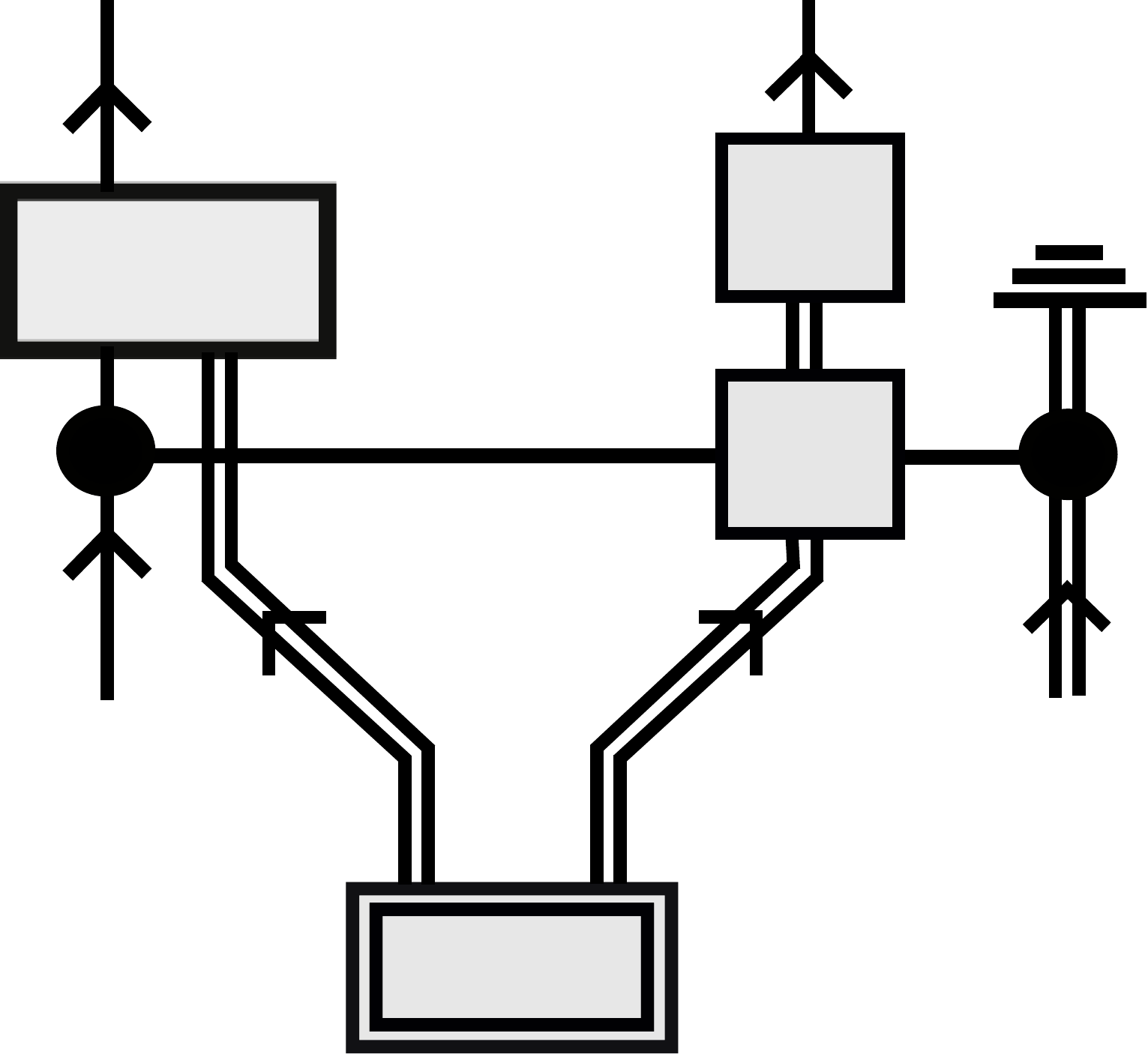}}
\put(-85,50){$x$}
\put(-85,115){$a$}
\put(-15,120){$B_{out}$}
\put(55,45){$B_{in}$}
\put(10,95){$N_b$}
\put(11,67){X}
\put(-26,7){$\rho_{AB}$}
\put(-26,-20){$(b)$}
}

  \end{center}
  \vspace*{-10mm}
  \caption{Mathematical depiction of two post-quantum MDI assemblages. (a) MDI assemblage $\N^{PTP}$: Alice and Bob share a Bell state; Alice performs measurements on her system, while Bob performs a controlled-transpose operation. (b) MDI assemblage $\N^{PR}$: Alice and Bob share a Bell state; Alice performs measurements on her system, while Bob performs a CNOT operation controlled by both Alice's and Bob's inputs.  }
\end{figure}

For our second example we define a post-quantum assemblage that generates PR box correlations between Alice and Bob. Mathematically it can be expressed as follows: 
\begin{align}\label{eq:PR-mdi}
\N^{PR}_{\A\B|\X} &= \left\{\Ne^{PR}_{ab|x}\right\}_{a\in \A, \, x \in \X, \, b \in \B} \,,\\
\text{with} \quad & \begin{cases} \Ne^{PR}_{ab|x}(\cdot) = \tr{(M_{a} \otimes \id_{B_{out}}) \, (\id_A \otimes \text{CX}_b^{BB_{in}\rightarrow B_{out}}) \,(\rho_{AB} \otimes (\cdot)_{B_{in}})}{} \,\text{if}\, x \in \{0,1\}, \nonumber \\
\Ne^{PR}_{ab|x}(\cdot) = \frac{1}{4}
\hskip22.8em\relax \text{if}\, x=2 \,, \nonumber \\
M_{0} = \ket{0}\bra{0}\,,\quad M_{1} = \ket{1}\bra{1}\,. \nonumber 
\end{cases}
\end{align}
Here, the mapping $\text{CX}_b^{BB_{in}\rightarrow B_{out}}$ is the same as in the channel EPR scenario, followed by a measurement in a computational basis. It is  illustrated in Fig.~\ref{fig:MDI-PR}. Notice that for $x \in \{0,1\}$, if Bob's input states are classical labels $\{\ket{y}\}_{y\in\{0,1\}}$, Alice and Bob obtain correlations $p(ab|xy)$ that correspond to the PR box. For simplicity, hereon we denote $\N^{PR}_{\A|\X\Y}=\N^{PR}$.

To study the pre-order of $\N^{PR}$ and $\N^{PTP}$, we convert the assemblages into Choi form and run the SDP~\ref{SDP-MDI} (in Matlab \cite{MATLAB:2010}, using the software CVX \cite{grant2013cvx,blondel2008recent}, the solver SeDuMi \cite{sturm1999using} and the toolbox QETLAB \cite{qetlab}; see the code at \cite{github}). We observe the following:

\begin{obs}\label{obs:MDI-PTP-PR}
The post-quantum measurement-device-independent assemblages $\N^{PR}$ and $\N^{PTP}$ are incomparable in the LOSR resource theory of common-cause assemblages. 
\end{obs}

\section{Related prior work}\label{se:relwork}

\subsection{Channel EPR scenarios}\label{sec:related-channel}

Ref.~\cite{piani2015channel}, which first introduced channel EPR scenarios, takes `channel steering' to be an instance of the resource theory of local operations and classical communication (LOCC), rather than the LOSR resource theory we have developed here. Both of these approaches lead to a meaningful resource theory: the LOSR approach is suited to the study of nonclassicality in scenarios where there are no cause-effect influences between parties, whereas the LOCC approach is suited to the study of scenarios in which there {\em are} cause-effect influences between parties. In our view, EPR scenarios do not involve causal influences from one party to the other, as argued in detail in Ref.~\cite{EPRLOSR}; hence, we consider the LOSR approach to be better suited to it.

The biggest distinction between the LOSR and LOCC approaches is that all resources in the LOSR approach are no-signalling by construction, while in the LOCC approach even free resources might be signalling from Bob to Alice. In Ref.~\cite{piani2015channel}, the classical channel assemblages are defined as the ones that admit a decomposition $\widetilde{\I}_{a|x}(\cdot)=\sum_{\lambda} p(a|x,\lambda) \widetilde{\I}_{\lambda}(\cdot)$, where $\widetilde{\I}_{\lambda}(\cdot)$ is a CP map that does not need to be trace preserving. In general, this definition of a free set of channel assemblages does not coincide with our definition specified in Eq.~\eqref{eq:free-channel}. However, when one restricts the scenario to non-signalling channel assemblages, these two approaches coincide, as we show in Appendix~\ref{app:free}.

When restricting the study of channel assemblages to that of no-signalling channel assemblages, the relationship between the LOSR and LOCC resource-theoretic approaches is analogous to the relationship between the traditional notion of LOCC-entanglement and the notion of LOSR-entanglement introduced in Refs.~\cite{buscemi2012all,schmid2020standard}. Indeed, just as for the LOCC and LOSR resource theories of entanglement (where the definitions of 
free resources---i.e., separable states---coincide), the sets of free assemblages for the LOCC and LOSR approaches also coincide. However, the different choice of free operations---LOCC or LOSR---{\em does} impact the relative ordering of assemblages. 

This leads to significant differences between the LOCC and LOSR approaches. For example, Ref.~\cite{banacki2022multipartite} shows that when one allows signalling from Bob to Alice, then all bipartite channel assemblages admit a quantum realization. In the LOSR approach this is not true - throughout the paper we studied many examples of bipartite post-quantum assemblages in channel, Bob-with-input and measurement-device-independent scenarios.

Ref.~\cite{piani2015channel} attempts to give a few arguments in favor of the LOCC approach over the LOSR approach to channel EPR scenarios; e.g., by noting that the LOCC set of free channel assemblages is elegantly characterized by properties of their Choi-Jamiołkowski representation, and by noting that it recovers the separable states in the appropriate special case. However, all of the arguments given are equally true in the LOSR approach, and so do not discriminate between the two approaches.

Lastly, let us comment that neither LOCC or LOSR is the largest set of transformations leaving the free objects invariant in the resource theory of channel assemblages. We leave the question of what this set is for future work, focusing only on physically motivated free operations in this paper \footnote{One example of an operation outside the LOSR set that leaves the free resources free is the following. Consider a Bob-with-input assemblage. First, swap Alice and Bob’s subsystems. Next, take what is now Alice’s quantum system and convert it to a classical output (by measuring in a preferred basis). On Bob's side, convert his classical system to a quantum one by conditioning a quantum state preparation on the value of the classical variable. This operation will not create nonclassicality, hence it could be considered a free operation. However, it does not seem to have any meaningful physical motivation.}.

\subsection{The resource theory of Local Operations and Shared Entanglement}

Another relevant resource theory with which we can compare our approach to is that of Local Operations and Shared Entanglement (LOSE) \cite{schmid2021postquantum}.  Ref.~\cite{schmid2021postquantum} considers various common-cause non-signalling resources, including assemblages. The free operations of this resource theory differ from LOSR only in that the arbitrary shared randomness (classical common cause) is replaced with arbitrary shared entanglement (quantum common cause). Just as the set of LOSR operations is a subset of LOCC operations, LOSR is also clearly a subset of LOSE. 

Rather than nonclassicality of common causes, the LOSE resource theory studies {\em postquantumness} of a common cause as a resource, and it allows some nonclassical common causes (the quantum ones) for free. 
Unsurprisingly, this implies that the pre-orders under LOSR and LOSE operations are different; e.g., in the latter, all quantum-realisable assemblages are freely interconvertible. 
The interesting differences between the pre-orders, then, will arise for post-quantum channel assemblages. 
Although we are far from a complete understanding of these differences, 
 we can here already comment on one. 
Recall that here we showed that under LOSR operations the post-quantum assemblages $\As^{PTP}$ and $\As^{PR}$ are incomparable resources in a Bob-with-input EPR scenario. From the viewpoint of LOSE, the situation however changes. 
In Ref.~\cite{schmid2021postquantum} it was shown that the Bob-with-input assemblage $\As^{PTP}$ can be generated from the PR box resource using LOSE operations. This  highlights  the power of quantum entanglement even when considering resources that are post-quantum.

\section*{Outlook}
\addcontentsline{toc}{section}{Outlook}

In this work we have fleshed out the details of a resource theory of channel assemblages in channel EPR scenarios under local operations and shared randomness as free operations. We have explored the general case of channel EPR scenarios, and also the particular ones that come from changing the system type of Bob's input and output systems, denoted by Bob-with-input EPR scenarios and measurement-device-independent EPR scenarios. In all cases we specified a semidefinite program that tests resource conversion under LOSR operations, and found curious properties of the resource pre-order. 

Looking forward, there are plenty of questions one may tackle.
In this paper, we approached the problem of characterizing the preorder of assemblages using SDPs. Indeed, the SDPs we derive can in principle give {\em complete} information about the preorder of resources. One relevant question from a quantum information perspective is how to leverage this framework to define resource monotones that can quantify how useful an assemblage is as a resource for specific quantum information processing and communication tasks. From a foundational perspective, there are also various curiosities one may pursue. One pertains to how the pre-order of post-quantum resources changes from one scenario to another. For example, one can map a Bob-with-input assemblage to a correlation in a Bell scenario by performing a measurement on Bob's system: could there exist a map (given by a fixed measurement) that takes two incomparable assemblages into the same post-quantum correlation? How can we phrase and pursue the question for the cases of channel EPR scenarios and MDI EPR scenarios? What valuable insights about nature can such EPR scenarios give us? 

\section*{Acknowledgments}
BZ thanks Michał Banacki for useful discussions. BZ, DS, and ABS acknowledge support by the Foundation for Polish Science (IRAP project, ICTQT, contract no. 2018/MAB/5, co-financed by EU within Smart Growth Operational Programme). MJH and ABS acknowledge the FQXi large grant ``The Emergence of Agents from Causal Order'' (FQXi FFF Grant number FQXi-RFP-1803B). This research was supported by Perimeter Institute for Theoretical Physics. Research at Perimeter Institute is supported in part by the Government of Canada through the Department of Innovation, Science and Economic Development Canada
and by the Province of Ontario through the Ministry of Colleges and Universities. BZ acknowledges partial support by the National Science Centre, Poland 2021/41/N/ST2/02242. For the purpose of Open Access, the author has applied a CC-BY public copyright license to any Author Accepted Manuscript (AAM) version arising from this submission. The diagrams within this manuscript were prepared using Inkscape.

\bibliography{LOSRsteeringreferences}
\bibliographystyle{quantum}

\appendix

\section{LOSR transformations in terms of deterministic combs}\label{app:deterministic}

\subsection{The channel EPR scenario}\label{app:deterministic-channel}
The most general LOSR operation on a channel assemblage transforms $\In_{\A|\X}$ into $\In'_{\A'|\X'}$ as specified in Eq.~\eqref{eq:LOSRtrans-channel}, recalled below: 
\begin{align*}
\begin{split}
\I'_{a'|x'}(\cdot)=\sum_{\lambda}  \sum_{a,x} 
p(a^{\prime},x|a,x^{\prime},\lambda)\, p(\lambda)\, \Lambda_\lambda^{B_{out}S \rightarrow B_{out'}'} \,(\I_{a|x}\otimes \id_S)\, \Lambda_{\lambda}^{B_{in'}' \rightarrow B_{in}S} (\cdot)\,.
\end{split}
\end{align*}
It was shown in Refs.~\cite{FinePRL,cowpie} that any indeterminism in Alice's local comb can be
absorbed into the shared common cause $\lambda$. That is, the indeterministic probability distribution $p(a^{\prime},x|a,x^{\prime},\lambda)$ is as general  as a deterministic one is.  In this subsection we recall the proof of this statement. 

First, decompose Alice's local comb as a convex combination of deterministic combs:
\begin{align}\label{eq:det-bipartite}
    p(a^{\prime},x|a,x^{\prime}, \lambda) = \sum_{\widetilde{\lambda}} p(\widetilde{\lambda}|\lambda) D(a^{\prime},x|a,x^{\prime},\widetilde{\lambda}),
\end{align}
with $D(a^{\prime},x|a,x^{\prime},\widetilde{\lambda})$ being a deterministic probability distribution. We can always express $D(a^{\prime},x|a,x^{\prime},\widetilde{\lambda})=D(a^{\prime}|a,x^{\prime},\widetilde{\lambda})D(x|a,x^{\prime},\widetilde{\lambda})$. Moreover, since $D(a^{\prime},x|a,x^{\prime},\widetilde{\lambda})$ satisfies the condition of no-retrocausation (the variable $a$ is the causal future of the variable $x$, therefore $a$ cannot influence the value of $x$), without loss of generality $D(x|a,x^{\prime},\widetilde{\lambda})=D(x|x^{\prime},\widetilde{\lambda}).$ Hence 
\begin{align}\label{eq:det-bipartite-NR}
D(a^{\prime},x|a,x^{\prime},\widetilde{\lambda})=D(a^{\prime}|a,x^{\prime},\widetilde{\lambda})\,D(x|x^{\prime},\widetilde{\lambda}).
\end{align}
Here, $D(a^{\prime}|a,x^{\prime},\widetilde{\lambda})$ assigns a fixed outcome $a^{\prime}$ for each possible choice of $a$, $x^{\prime}$, and $\widetilde{\lambda}$, 
and $D(x|x^{\prime},\widetilde{\lambda})$ assigns a fixed outcome $x$ for each measurement $x^{\prime}$ and value of $\widetilde{\lambda}$.

Putting this together in the Choi form with the definition of the map $J_{\cE\, \lambda}^{\,B_{in'}'B_{out} \rightarrow B_{in}B_{out'}'}$ (after Eq.~\eqref{eq:LOSRtrans-channel}) one obtains: 
\begin{align*}
 J'_{a'|x'}(\cdot)=\sum_{\lambda,\widetilde{\lambda}}  \sum_{a,x} 
p(\widetilde{\lambda}|\lambda) \,D(a^{\prime}|a,x^{\prime},\widetilde{\lambda})\,D(x|x^{\prime},\widetilde{\lambda}) \, J_{a|x}\, * J_{\cE\, \lambda} . 
\end{align*}
From here, one can then define 
\begin{align}
 \widetilde{J}_{\cE\, \widetilde{\lambda}}^{\,B_{in'}'B_{out} \rightarrow B_{in}B_{out'}'} (\cdot)\,= \sum_{\lambda} p(\widetilde{\lambda}|\lambda) \,  J_{\cE\, \lambda}^{\,B_{in'}'B_{out} \rightarrow B_{in}B_{out'}'} (\cdot)\,, 
\end{align}
where $\widetilde{J}_{\cE\, \lambda}^{\,B_{in'}'B_{out} \rightarrow B_{in}B_{out'}'}$ is by definition a map that is non-signalling from the system defined on $\cH_{B_{out}}$ to the system defined on $\cH_{B_{in}}$. From here follows that: 
\begin{align}
J'_{a'|x'}(\cdot)=\sum_{\widetilde{\lambda}}  \sum_{a,x} 
 D(a^{\prime}|a,x^{\prime},\widetilde{\lambda})\,D(x|x^{\prime},\widetilde{\lambda}) \, J_{a|x}\, * \widetilde{J}_{\cE\, \widetilde{\lambda}}\,,
\end{align}
which gives Eq.~\eqref{eq:LOSR-for-SDP} from the main text.

\subsection{The Bob-with-input EPR scenario}\label{app:deterministic-BwI}

The arguments given above also apply to the Bob-with-input EPR scenario. The most general LOSR transformation on a Bob-with-input assemblage transforms $\As_{\A|\X\Y}$ into $\As'_{\A'|\X'\Y'}$ as follows:
\begin{align}
\begin{split}
\sigma'_{a'|x'y'}=\sum_{\lambda}  \sum_{a,x,y} 
p(a^{\prime},x|a,x^{\prime},\lambda) p(y|y^{\prime},\lambda) p(\lambda) \cE_{\lambda, y'}(\sigma_{a|xy}).
\end{split}
\end{align}

We already discussed that Alice's local comb can be expressed in terms of deterministic probability distributions as in Eq.\eqref{eq:det-bipartite}. The same technique can be applied to Bob's pre-processing to express it as
\begin{equation}
    p(y|y^{\prime},\lambda)= \sum_{\widetilde{\lambda}} p(\widetilde{\lambda}|\lambda) D(y|y^{\prime},\widetilde{\lambda}).
\end{equation}
Define the following CPTP map:
\begin{equation}
    \widetilde{\cE}_{\widetilde{\lambda},y^{\prime}}(\sigma_{a|xy})=\sum_{\lambda} p(\widetilde{\lambda}|\lambda) p(\lambda) \cE_{\lambda,y'}(\sigma_{a|xy}).
\end{equation}
We can now rewrite the Bob-with-input LOSR transformation as
\begin{align}
\sigma'_{a'|x'y'}=\sum_{\widetilde{\lambda}}  \sum_{a,x,y} D(x|x^{\prime},\widetilde{\lambda}) 
D(a^{\prime}|a,x^{\prime},\widetilde{\lambda}) D(y|y^{\prime},\widetilde{\lambda})  \widetilde{\cE}_{\widetilde{\lambda},y^{\prime}}(\sigma_{a|xy}),
\end{align}
which is exactly Eq.~\eqref{eq:LOSRtrans-BWI2} from the main text if one labels $\widetilde{\lambda}$ as $\lambda$.

\section{Robust formulation of the SDPs}\label{app:SDP}

All SDPs derived in this paper are feasibility problems, i.e., they are written in a form where the objective function vanishes. One can relax the constraints on the Choi matricies and instead of requiring their positive semi-definiteness, require them to be \textit{close} to a positive semi-definite matrix. Such reformulation makes the SDPs more robust. In SDP~\ref{SDP-channel}, this can be implemented by relaxing the constraints $J_{\cE\,\lambda} \geq 0$ and $J_{F\,\lambda} \geq 0$ to $J_{\cE\,\lambda} + \mu \id \geq 0$ and $J_{F\,\lambda} + \mu \id \geq 0$, respectively. Then, the SDP~\ref{SDP-channel} can be written as a minimization of the new parameter $\mu$:

\begin{sdp}\label{SDP-channel-robust}
The channel assemblage $\In_{\A|\X}$ can be converted to the channel assemblage $\In'_{\A'|\X'}$ under LOSR operations, denoted by $\In_{\A|\X} \, \overset{\text{LOSR}}{\longrightarrow} \, \In'_{\A^{\prime}|\X^{\prime}}$, if and only the solution of the following SDP satisfies $\mu < 10^{-10}$:
\begin{align}
\begin{split}
\textrm{given} \;\;\;\;\;& \{ J_{a|x}\}_{a,x}\,,\; \{J'_{a'|x'}\}_{a',x'}\,,\; \{D(a^{\prime}|a,x^{\prime},\lambda)\}_{\lambda,a^{\prime},a,x^{\prime}} \,, \;  \{D(x|x^{\prime},\lambda)\}_{\lambda,x,x^{\prime}} \\
\underset{(J_{\cE\,\lambda})_{\lambda} \,,\;(J_{F\,\lambda})_{\lambda}}{\textrm{min}} &  \;\;\;\;\; \mu  \\
    \textrm{s.t.} \;\;\;& \begin{cases} 
    \mu \geq 0\,, \\
   J_{\cE\,\lambda} + \mu \id \geq 0 \quad \forall \lambda\,,\\
      \tr{J_{\cE\,\lambda}}{B_{out'}'B_{in}} \propto \id_{B_{out}B_{in'}'} \quad \forall \lambda\,, \\
      \sum_{\lambda} \tr{J_{\cE\,\lambda}}{B_{out'}'B_{in}} = \frac{1}{d_{B_{out}}d_{B_{in'}'}}\, \id_{B_{out}B_{in'}'}\,, \\
          J_{F\,\lambda} + \mu \id \geq 0 \quad \forall \lambda\,,\\
      \tr{J_{F\,\lambda}}{B_{in}} \propto \id_{B_{in'}'} \;\;\;\; \forall \lambda\,, \\
      \sum_{\lambda} \tr{J_{F\,\lambda}}{B_{in}} = \frac{1}{d_{B_{in'}'}}\, \id_{B_{in'}'}\,, \\
      \tr{J_{\cE\,\lambda}}{B_{out'}'} = J_{F\,\lambda} \otimes \frac{1}{d_{B_{out}}}\,\id_{B_{out}} \quad \forall \lambda \,, \\
            J'_{a'|x'}=\sum_{\lambda}  \sum_{a,x} 
D(a^{\prime}|a,x^{\prime},\lambda)\,D(x|x^{\prime},\lambda)\,   J_{a|x}\, * J_{\cE\,\lambda} \,.
      \end{cases}
    \end{split}
\end{align}
\end{sdp}
The threshold for the value of $\mu$ can be selected depending on the application of the program. Above, we decide on $\mu < 10^{-10}$ solely for the purpose of the presentation. 

This method can be applied to the SDP~\ref{sdp:BWI} and SDP~\ref{SDP-MDI} in an analogous way, hence we do not present the modified SDPs here.

We run the SDP~14 (and analogous robust SDPs for the Bob-with-input and measurement-device-independent scenarios) to confirm results stated in the paper. For conversions which are possible, the result of the optimization is $\mu \approx 10^{-10}$. For impossible conversions, it is $\mu > 10^{-4}$. The code is available at \cite{github}.

\section{Almost quantum correlations}\label{app:AQ}

In this Appendix, we recall an example from Ref.~\cite{AQ} of a probability distribution that is almost-quantum  yet not quantum.  Consider a bipartite Bell scenario where $a,b,x,$ and $y$ are binary classical variables. In this scenario, the probability distribution is fully specified by an 8-dimensional vector $\vec{p}$ with the following entries:
\begin{align}
    \vec{p} =\Big\{ p_A(1|0), p_A(1|1), p_B(1|0), p_B(1|1), p(1, 1|0, 0), p(1, 1|1, 0), p(1, 1|0, 1), p(1, 1|1, 1) \Big\}.
\end{align}
Here, $p_A(a|x)$ and $p_B(b|y)$ are  the marginal  probabilities  corresponding  to Alice's and Bob's  individual  measurements, respectively, and $p(ab|xy)$ is the conditional joint probability distribution. To calculate the probabilities that are not explicitly contained in $\vec{p}$, one must simply use
the normalization and no-signalling constraints:
\begin{align}
     \sum_{a \in \A, b \in B} p(ab|xy) &= 1 \quad \forall \, x,y, \\
     \sum_{b \in B} p(ab|xy) &=  p_A(a|x)  \quad \forall \, x,y,a, \\
     \sum_{a \in A} p(ab|xy) &=  p_B(b|y)  \quad \forall \, x,y,b.
\end{align}

In Ref.~\cite{AQ}, it was shown that the following probability vector belongs to the almost quantum set and it does not admit a quantum realization:
\begin{align}
    \vec{p}_{AQ} =\Big\{ \frac{9}{20}, \frac{2}{11}, \frac{2}{11}, \frac{9}{20}, \frac{22}{125}, \frac{\sqrt{2}}{9}, \frac{37}{700}, \frac{22}{125} \Big\}.
\end{align}
In the main text, we use the probability $\vec{p}_{AQ}$ to construct the elements of the Bob-with-input assemblage $\As^{AQ}$.

\section{Membership problem for the measurement-device-independent EPR scenario}\label{app:MDI} 

Deciding whether a measurement-device-independent assemblage has a quantum realization is a highly non-trivial problem. In this section, we introduce the first level of a hierarchy of semidefinite programs for MDI EPR scenarios which tests a membership of the set of quantum assemblages.  The  hierarchy of programs will be presented in future work; for the purpose of this paper, we only specify its first level. This hierarchy is analogous to the Navascu{\'e}s-Pironio-Ac{\'\i}n (NPA) hierarchy for quantum correlations \cite{navascues2007bounding,navascues2008convergent}.

The elements of a quantum MDI assemblage $\N_{\A\B|\X}=~\{\Ne_{ab|x}(\cdot)\}$ can be expressed as
\begin{eqnarray*}
\Ne_{ab|x}(\cdot)=\tr{(M_{a|x}^{A} \otimes \Theta_b^{BB_{in} \rightarrow B_{out}})[\rho_{AB} \otimes (\cdot) _{B_{in}}]
}{}.
\end{eqnarray*}
Without loss of generality, we can take the shared state to be pure, which we denote by $\ket{\psi}\bra{\psi}_{AB}$, and the measurement $\Theta_b$ to be projective, which we denote by $F_{b}^{BB_{in}}$. Thus, each element $\Ne_{ab|x}$ of the assemblage can be written as $\Ne_{ab|x}(\cdot)=\textrm{tr}\left(\widetilde{\Ne}_{ab|x}(\cdot)\right)$, with
\begin{eqnarray*}
\widetilde{\Ne}_{ab|x}  (\cdot)  &=&\textrm{tr}_{AB}\left(M_{a|x}^{A} \otimes F_{b}^{BB_{in}}\,|\psi\rangle\langle\psi|_{AB}  \otimes (\cdot) \right)\\
&=&\langle\psi|M_{a|x}^{A} \otimes F_{b}^{BB_{in}}|\psi\rangle  (\cdot)  .
\end{eqnarray*}

Notice that $\widetilde{\Ne}_{ab|x}$ is a matrix and we can write its elements as
\begin{eqnarray*}
\langle i |\widetilde{\Ne}_{ab|x}| j\rangle&=&\langle i |\langle\psi|M_{a|x}\otimes F_{b}|\psi\rangle| j\rangle\\
&=&\langle\psi|M_{a|x}\otimes \langle i|F_{b}|j\rangle|\psi\rangle\\
&=&\langle\psi|M_{a|x}\otimes F^{ij}_{b}|\psi\rangle.
\end{eqnarray*}
Here, the operator $F^{ij}_{b}:= \langle i|F_{b}|j\rangle$ must satisfy the following properties:
\begin{eqnarray*}\label{constraints}
\sum_{b}F_{b}^{ij}=\delta_{i,j}\mathbb{I}_{B}
\,,&\quad  &
\sum_{i}(F_{b}^{ij})^{\dagger}F_{b'}^{ij'}=\delta_{b,b'}F_{b}^{jj'}.
\end{eqnarray*}
Therefore, all elements of the matrix $\widetilde{\Ne}_{ab|x}$ can be written as linear combinations of inner products of the form $\langle\psi|O_{k}^{\dagger}O_{m}|\psi\rangle$, where $O_{k}\in\{\mathbb{I}, \{M_{a|x}\}, \{F_{b}^{ij}\}\}$. We can now write a moment matrix with columns and rows indexed by the choice of operator $O_k$. Denote $\langle\psi|O_{k}^{\dagger}O_{m}|\psi\rangle$ as $\langle O_{k}^{\dagger}O_{m}\rangle$. The moment matrix for $a,b \in \{0, 1\}$, $x\in\{0,1,2\}$ and $B_{in}$ being two-dimensional is the following:
\[
\begin{bmatrix}
\langle\psi|\psi\rangle &  \langle M_{0|0}\rangle &  \langle M_{0|1}\rangle &  \langle M_{0|2}\rangle &  \langle F_{0}^{00}\rangle &  \langle F_{0}^{01}\rangle &  \langle F_{0}^{10}\rangle &  \langle F_{0}^{11}\rangle\\

  \langle M_{0|0}\rangle&  \langle M_{0|0}M_{0|0}\rangle &  \langle M_{0|0}M_{0|1}\rangle &  \langle M_{0|0}M_{0|2}\rangle   &  \langle M_{0|0}F_{0}^{00}\rangle &  \langle M_{0|0}F_{0}^{01}\rangle &  \langle M_{0|0}F_{0}^{10}\rangle &  \langle M_{0|0}F_{0}^{11}\rangle\\
 
   \langle M_{0|1}\rangle & \langle M_{0|1}M_{0|0}\rangle&  \langle M_{0|1}M_{0|1}\rangle &  \langle M_{0|1}M_{0|2}\rangle &   \langle M_{0|1}F_{0}^{00}\rangle &  \langle M_{0|1}F_{0}^{01}\rangle &  \langle M_{0|1}F_{0}^{10}\rangle &  \langle M_{0|1}F_{0}^{11}\rangle\\
  
     \langle M_{0|2}\rangle &  \langle M_{0|2}M_{0|0}\rangle&   \langle M_{0|2}M_{0|1}\rangle&  \langle M_{0|2}M_{0|2}\rangle &    \langle M_{0|2}F_{0}^{00}\rangle &  \langle M_{0|2}F_{0}^{01}\rangle &  \langle M_{0|2}F_{0}^{10}\rangle &  \langle M_{0|2}F_{0}^{11}\rangle\\
    
       \langle F_{0}^{00}\rangle    & \langle F_{0}^{00}M_{0|0}\rangle & \langle F_{0}^{00}M_{0|1}\rangle &\langle F_{0}^{00}M_{0|2}\rangle  &    \langle F_{0}^{00}F_{0}^{00}\rangle &  \langle F_{0}^{00}F_{0}^{01}\rangle &  \langle F_{0}^{00}F_{0}^{10}\rangle &  \langle F_{0}^{00}F_{0}^{11}\rangle\\
                
	\langle F_{0}^{10}\rangle  & \langle F_{0}^{10}M_{0|0}\rangle& \langle F_{0}^{10}M_{0|1}\rangle & \langle F_{0}^{10}M_{0|2}\rangle & \langle F_{0}^{10}F_{0}^{00}\rangle &    \langle F_{0}^{10}F_{0}^{01}\rangle &  \langle F_{0}^{10}F_{0}^{10}\rangle &  \langle F_{0}^{10}F_{0}^{11}\rangle\\
 
\langle F_{0}^{01}\rangle  & \langle F_{0}^{01}M_{0|0}\rangle & \langle F_{0}^{01}M_{0|1}\rangle & \langle F_{0}^{01}M_{0|2}\rangle & \langle F_{0}^{01}F_{0}^{00}\rangle & \langle F_{0}^{01}F_{0}^{01}\rangle &    \langle F_{0}^{01}F_{0}^{10}\rangle &  \langle F_{0}^{01}F_{0}^{11}\rangle\\

\langle F_{0}^{11}\rangle  & \langle F_{0}^{11}M_{0|0}\rangle & \langle F_{0}^{11}M_{0|1}\rangle & \langle F_{0}^{11}M_{0|2}\rangle & \langle F_{0}^{11}F_{0}^{00}\rangle & \langle F_{0}^{11}F_{0}^{01}\rangle &    \langle F_{0}^{11}F_{0}^{10}\rangle &    \langle F_{0}^{11}F_{0}^{11}\rangle\\
\end{bmatrix}.
\]
This moment matrix  can be seen to specify  the first level of a hierarchy of semidefinite programs that test membership of an MDI assemblage in the quantum set. Higher levels of the hierarchy can be generated by considering various sequences of products of operators $M_{a|x}$ and $F_{b}^{ij}$,  as we will specify in a follow up manuscript. 

We used the first level of the hierarchy to check if the MDI assemblage specified in Eq.~\eqref{eq:PTP-MDI} is quantum. We run the SDP (in Matlab \cite{MATLAB:2010}, using the software CVX \cite{grant2013cvx,blondel2008recent}; see the code at \cite{github}) to check if any moment matrix of the form above is positive semidefinite, and it is not. We conclude that the assemblage is postquantum.

\section{Proofs}\label{app:proofs}
\subsection{Proof of Theorem~\ref{thm:PR-PTP}}\label{app:proof-thm}

In this section, we give a proof of Theorem~\ref{thm:PR-PTP} stated in the main text. To show that $\As^{PR}$ cannot be converted into $\As^{PTP}$ with LOSR operations, we will make use of the `steering' functional constructed in Ref.~\cite[Eq.~(D3)]{sainz2020bipartite}, which we denote $S_{PTP}$. A generic functional in a Bob-with-input scenario is defined as: 
\begin{align}
S[\As_{\A|\X\Y}] = \Tr{ \sum_{a\in \A, \, x\in \X,\, y\in \Y} F_{axy} \, \sigma_{a|xy}}\,,
\end{align}
where $\{ F_{axy} \}$ is a set of Hermitian operators. $S_{PTP}$ is specified by the following operators
\begin{align}
F^{PTP}_{axy}=\frac{1}{2} (\id - (-1)^a \sigma_x)^{T^{y}}\,,
\end{align}
where $T^y$ denotes that the transpose operation is applied when $y=1$, and the identity operation is applied when $y=0$.
 Here, however, the operator $\sigma_x$ should not be confused with the Pauli-X operator, since $x \in \X$ here denotes the choice of Alice's measurement. Indeed, the operators $\sigma_0$, $\sigma_1$, and $\sigma_2$ (since $x \in \{0,1,2\}$) should be thought of as the Pauli X, Y, and Z operators respectively. This alternative way of denoting the Pauli operators will also be used later on in this section. 

It was shown in Ref.~\cite[Appendix~D]{sainz2020bipartite} that for any non-signalling assemblage $\As_{\A|\X\Y}$, the minimum value $S_{PTP}^{\min}$ of $S_{PTP}[\As_{\A|\X\Y}]$ is given by $S_{PTP}^{\min}=0$. It was also proved that $\As^{PTP}$ achieves the minimum value, that is, $S_{PTP}[\As^{PTP}] = S_{PTP}^{\min}$.

It can be shown by direct calculation that $\As^{PR}$ does not achieve the value of $0$ for $S_{PTP}$. 
In order to show, hence, that $\As^{PR}$ cannot be converted to $\As^{PTP}$, 
it then suffices to show that any LOSR-processing of $\As^{PR}$ will also not give this minimum value of $0$ for $S_{PTP}$. To show this we will use the following observation.

\begin{rem}\label{rem:orth}
An assemblage $\As_{\A|\X\Y}$ satisfies $S_{PTP}[\As_{\A|\X\Y}] = S_{PTP}^{\min}=0$ iff each element satisfies $\sigma_{a|xy}=\alpha_{a,x,y}\, \frac{1}{2}\, (\id + (-1)^a \sigma_x)^{T^{y}}$ for all $a \in \A, x \in \X, y \in \Y$, where $\alpha_{a,x,y}$ is a real number such that $0\leq\alpha_{a,x,y}\leq 1$ for all $a \in \A, x \in \X, y \in \Y$. 
\end{rem}
\begin{proof}
First note that the operators $F^{PTP}_{axy}=\frac{1}{2} (\id - (-1)^a \sigma_x)^{T^{y}}$ are rank-1 projectors. Hence, $\Tr{ F^{PTP}_{axy} \, \sigma_{a|xy}}\geq 0 \quad$ for all$ \quad a\in \A, \, x\in \X,\, y\in \Y$. Since by assumption $S_{PTP}[\As_{\A|\X\Y}]=0$, it follows that $\Tr{ F^{PTP}_{axy} \, \sigma_{a|xy}} = 0\quad$ for all$ \quad a\in \A, \, x\in \X,\, y\in \Y$.
Since each $F^{PTP}_{axy}$ is a projector onto a one-dimensional subspace, this implies that $\sigma_{a|xy}$ cannot have any support on this one-dimensional subspace, and thus the assemblage element $\sigma_{a|xy}$ only has support on the one-dimensional subspace orthogonal to $F^{PTP}_{axy}$.
In other words, if $F^{PTP}_{axy}=\frac{1}{2} (\id - (-1)^a \sigma_x)^{T^{y}}$, then the assemblage element has support on the one-dimensional subspace spanned by the projector $\frac{1}{2} (\id + (-1)^a \sigma_x)^{T^{y}}$, which concludes the proof.
\end{proof}
Remark \ref{rem:orth} is crucial for proving Theorem~\ref{thm:PR-PTP}, which we recall here: 
\setcounter{theo}{7} 
\begin{thm}
$\As^{PR}$ cannot be converted into $\As^{PTP}$ with LOSR operations. 
\end{thm}

\setcounter{theo}{14}

\begin{proof}
Let us prove this by contradiction. Assume that $\As^{PTP}$ is given by an LOSR-processing of $\As^{PR}$, denoted as $\As^{PR}_{\Lambda} = \{\sigma^{PR,\Lambda}_{a|xy}\}$. Then, it follows that this LOSR-processing of $\As^{PR}$ saturates the non-signalling bound of the functional $S_{PTP}$. From Remark~\ref{rem:orth}, it can only be the case that each $\sigma^{PR,\Lambda}_{a|xy}$ is proportional to a particular Pauli eigenstate, where $x$ denotes its corresponding Pauli basis.
In other words, each element of $\As^{PR}_{\Lambda}$ is diagonal in a particular Pauli basis, and  has  one null eigenvalue.

Now note that every element of $\As^{PR}$ is diagonal in the computational (Pauli Z) basis. Therefore, the original assemblage $\As^{PR}$ is invariant under the dephasing map in the computational basis, i.e., for all $a\in \A, x\in \X, y \in \Y$, $D[\sigma^{PR}_{a|xy}]:=\sum_{i\in\{0,1\}}\proj{i}\sigma^{PR}_{a|xy}\proj{i}= \sigma_{a|xy}^{PR}$. Therefore, without loss of generality, any LOSR-processing of $\As^{PR}$ can be pre-composed with the dephasing map $D[\cdot]$ without changing the assemblage that results from the processing. 
More formally, for any LOSR transformation $\Lambda$, $\Lambda[D[\As^{PR}]]=\Lambda[\As^{PR}]=\As^{PR}_{\Lambda}$. Equivalently, \textit{any} LOSR-processing  of $\As^{PR}$  can be first described as an application of the dephasing map on Bob's qubit, following by another LOSR transformation. 

Therefore, for the elements of $\As^{PR}_{\Lambda}$ that are diagonal in Pauli bases other than the Pauli Z basis, the LOSR-processing needs to transform Bob's qubit so that the assemblage elements are diagonal in another basis. Furthermore, the elements of $\As^{PR}_{\Lambda}$ only have support in (at most) one-dimensional subspace. Such a situation for $\As^{PR}_{\Lambda}$ can only happen if, in the LOSR-processing on $\As^{PR}$, after the dephasing map, Bob applies another CPTP map $\mathcal{F}_{y,\lambda}$, which can depend on $y$ and the shared randomness $\lambda$. Because of the dephasing map, to produce assemblage elements with support on a one-dimensional subspace,  the map $\mathcal{F}_{y,\lambda}$ is such that it  can be simulated by a measure-and-prepare  channel, where  the preparation step in the channel should prepare the state corresponding to the one-dimensional subspace. However, the Pauli bases for the assemblage elements are necessarily determined by $x$, and, since $\mathcal{F}_{y,\lambda}$ cannot depend on $x$, it is impossible for an LOSR-processing to produce the desired assemblage $\As^{PR}_{\Lambda}$. That is, if the state preparation step of the channel $\mathcal{F}_{y,\lambda}$ could prepare states in the correct basis, then the measurement in the channel would reveal what the input $x$ of Alice is, which is incompatible with the no-signalling condition.
\end{proof}

\subsection{Proof of Observation~\ref{obs:BWI-AQ-PR}}\label{app:proofs-obs}
In Section~\ref{sec:preorder-BwI}, we used SDP~\ref{sdp:BWI} to certify that  $\As^{\prime PR}$  can be converted into $\As^{AQ}$ with LOSR operations, but not vice versa. We now prove the first part of this observation analytically. 
\begin{thm}\label{thm:BWI-AQ-PR}
The post-quantum Bob-with-input assemblage  $\As^{\prime PR}$  can be converted into $\As^{AQ}$ with LOSR operations. 
\end{thm}
\begin{proof}
 We prove the theorem by providing an explicit LOSR protocol that preforms the desired transformation. This protocol is depicted in Fig.~\ref{fig:protocol}. 
\begin{compactenum}
\item $\mathbf{\As^{\prime PR} \overset{\text{LOSR}}{\longrightarrow}} \vec{\mathbf{p}}_\mathbf{{PR}}$\textbf{:} Bob measures his share of the system in the computational basis. This local operation maps the assemblage $\As^{\prime PR}$ into a conditional probability distribution $\vec{p}_{PR} = \{p_{PR}(ab|xy)\}$ corresponding to PR-box correlations.
\item $\vec{\mathbf{p}}_\mathbf{{PR}}\mathbf{ \overset{\text{LOSR}}{\longrightarrow} }\vec{\mathbf{p}}_\mathbf{{AQ}}$\textbf{:} In the resource theory of common-cause boxes under LOSR operations, the equivalence class of PR-boxes is at the top of the pre-order of bipartite correlations with $|\A|=|\B|=|\X|=|\Y|=2$~\cite{cowpie}. Hence there exits an LOSR process that maps $\vec{p}_{PR}$ into $\vec{p}_{AQ}$. Let Alice and Bob apply this LOSR operation to $\vec{p}_{PR}$. They are then left sharing the correlations $\vec{p}_{AQ} = \{p_{AQ}(a'b'|x'y')\}$. 
\item $\vec{\mathbf{p}}_\mathbf{{AQ}}\mathbf{  \overset{\text{LOSR}}{\longrightarrow} \As^{AQ}}$\textbf{:} in this step, Bob implements a measure-and-prepare channel which reads out the value of the classical output $b'$ and prepares a qubit on state $\ket{b'}$. This state will be prepared with probability $p_{AQ}(a'b'|x'y')$ since Alice and Bob share $\vec{p}_{AQ}$. This hence effectively prepares $\As^{AQ}$. 
\end{compactenum}

\begin{figure}[h!]
  \begin{center}
\includegraphics[width=0.35\textwidth]{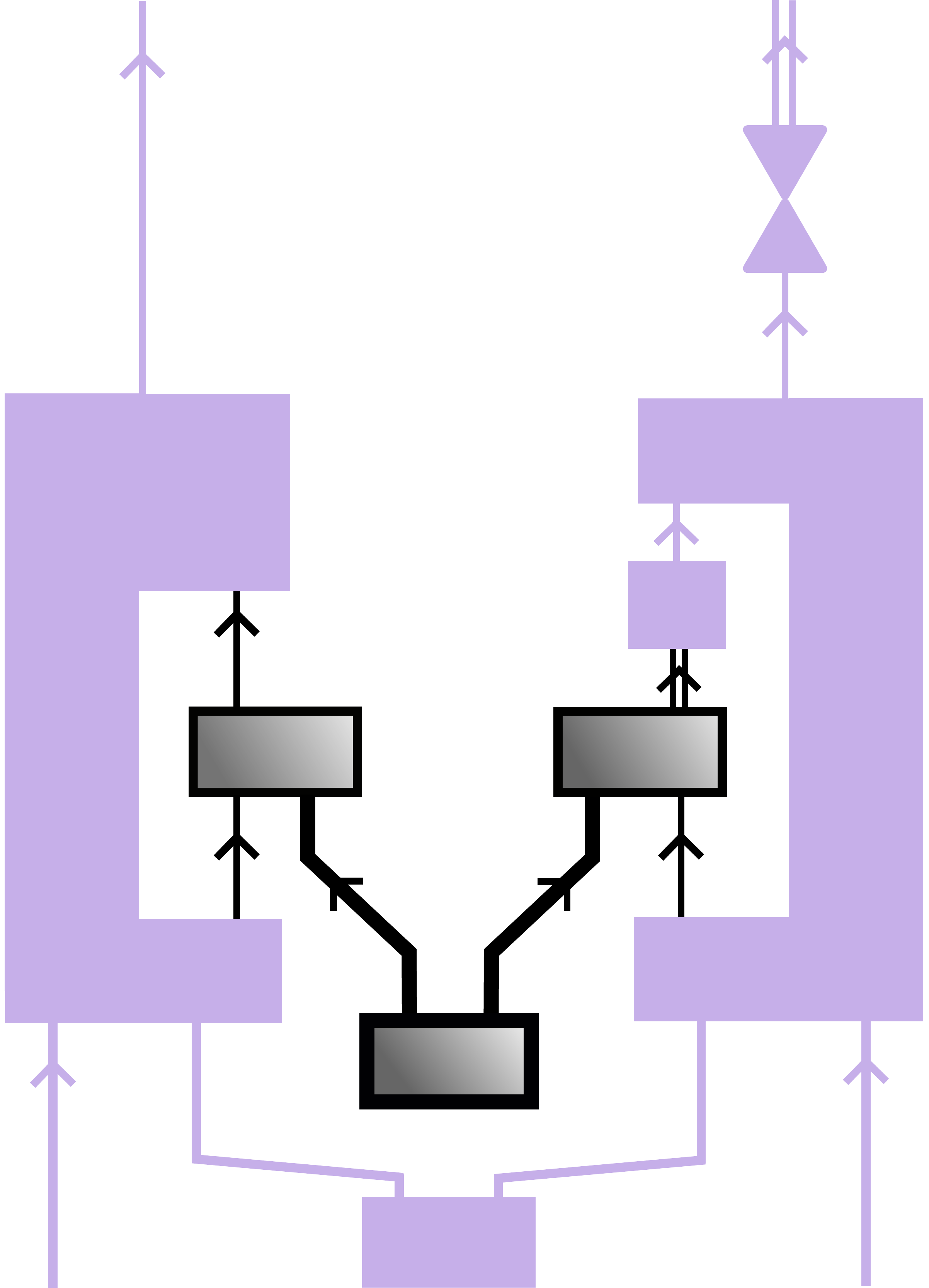}
\put(-130,70){$x$}
\put(-35,70){$y$}
\put(-130,105){$a$}
\put(-38,125){$b$}
\put(-74,105){$\rho_{a|xy}^{PR'}$}
\put(-85,5){$\lambda$}
\put(-170,20){$x'$}
\put(0,20){$y'$}
\put(-150,175){$a'$}
\put(-15,157){$b'$}
\put(-15,215){$\rho_{a|xy}^{AQ}$}

  \end{center}

  \caption{The LOSR protocol (in purple) that converts $\As^{\prime PR}$ into $\As^{AQ}$.
  }\label{fig:protocol}
\end{figure}
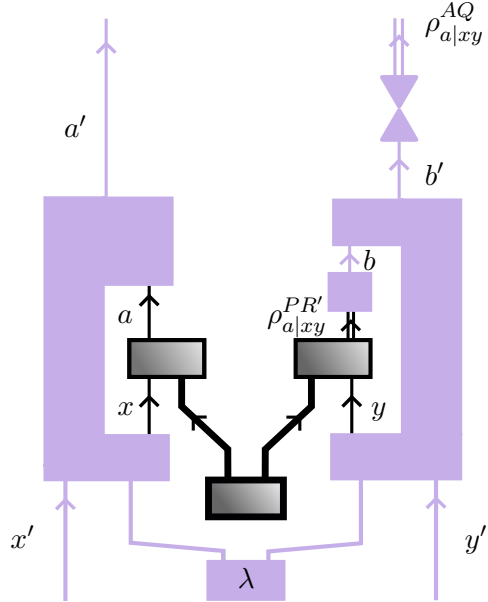

\end{proof}

\section{The sets of free non-signalling channel assemblages under LOCC and LOSR coincide}\label{app:free}

In our resource theory, we require the elements of a free channel assemblage to decompose as per Eq.~\eqref{eq:free-channel}, i.e. 
\begin{equation}
    \I_{a|x}(\cdot)=\sum_{\lambda} p(a|x,\lambda) p(\lambda) \I_{\lambda}(\cdot),
    \nonumber
\end{equation}
where $\I_{\lambda}(\cdot)$ is a CPTP map. It is straightforward to see that for \textit{any} state $\rho$, $\tr{\I_{a|x}(\rho)}{}=p(a|x)$, i.e., Alice's probability distribution $p(a|x)$ is independent of Bob's input state. This is the no-signalling condition from Bob to Alice, which is therefore satisfied in our LOSR approach. 

In Refs.~\cite{piani2015channel,banacki2022multipartite}, a free assemblage (see Fig.~\ref{fig:I1}) is defined as one that decomposes as
\begin{equation}
    \widetilde{\I}_{a|x}(\cdot)=\sum_{\lambda} p(a|x,\lambda) \widetilde{\I}_{\lambda}(\cdot),
    \nonumber
\end{equation}
with $\widetilde{\I}_{\lambda}$ being a CPTNI map. One can think about this construction as a channel that is semi-causal, i.e., non-signalling from Alice to Bob. In general,  however,  $\tr{\widetilde{\I}_{a|x}(\rho)}{}$ can depend on the particular choice of $\rho$,  and so it may allow signalling from Bob to Alice.

If one wants the free assemblage to be non-signalling (from Alice to Bob \textit{and} from Bob to Alice), an additional condition must be imposed on $\widetilde{\I}_{a|x}(\cdot)$. For this channel to be causal, it must also be semi-causal from Bob to Alice. Let us denote such assemblage by $\widetilde{\I}^{NS}_{a|x}(\cdot)$. Then, from the main result of Ref.~\cite{eggeling2002semicausal}, if follows that $\widetilde{\I}^{NS}_{a|x}(\cdot)$ must be semi-localizable from Bob to Alice. Therefore, $\widetilde{\I}^{NS}_{a|x}(\cdot)$ can be graphically represented as in Fig.~\ref{fig:I2}. Here, Bob's local operation is a CPTP map. The state that Alice and Bob now share is a classical-quantum separable state. For this reason, we can treat Bob's subsystem of the shared state, $\rho_\lambda$, as a quantum state that is being prepared conditioned on the value of the classical variable $\lambda$. This means that we can treat this preparation as a local operation on Bob's side, and the only system that he shares with Alice is the classical variable $\lambda$ as illustrated in  Fig.~\ref{fig:I3}.  It is clear to see that the assemblage illustrated in Fig.~\ref{fig:I2} is LOSR-free, i.e., it can be expressed as $\I_{a|x}(\cdot)=\sum_{\lambda} p(a|x,\lambda) p(\lambda) \I_{\lambda}(\cdot)$, with $\I_{\lambda}(\cdot)$ being a CPTP map. Therefore, we showed that the set of LOSR-free channel assemblages coincides with the set of free channel assemblages introduced in Refs.~\cite{piani2015channel,banacki2022multipartite} if no-signalling from Bob to Alice is imposed.

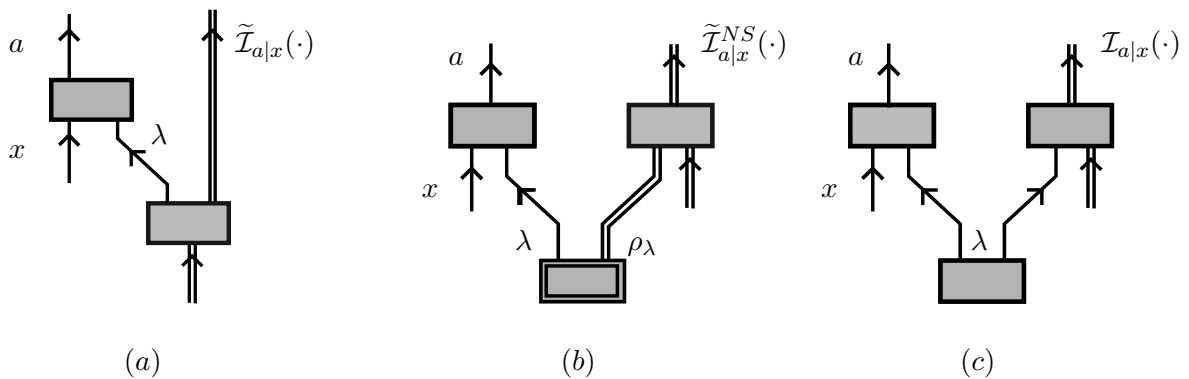
\begin{figure}[h!]
  \begin{center}
  \subcaptionbox{\label{fig:I1}}
{\put(-35,0){\includegraphics[width=0.15\textwidth]{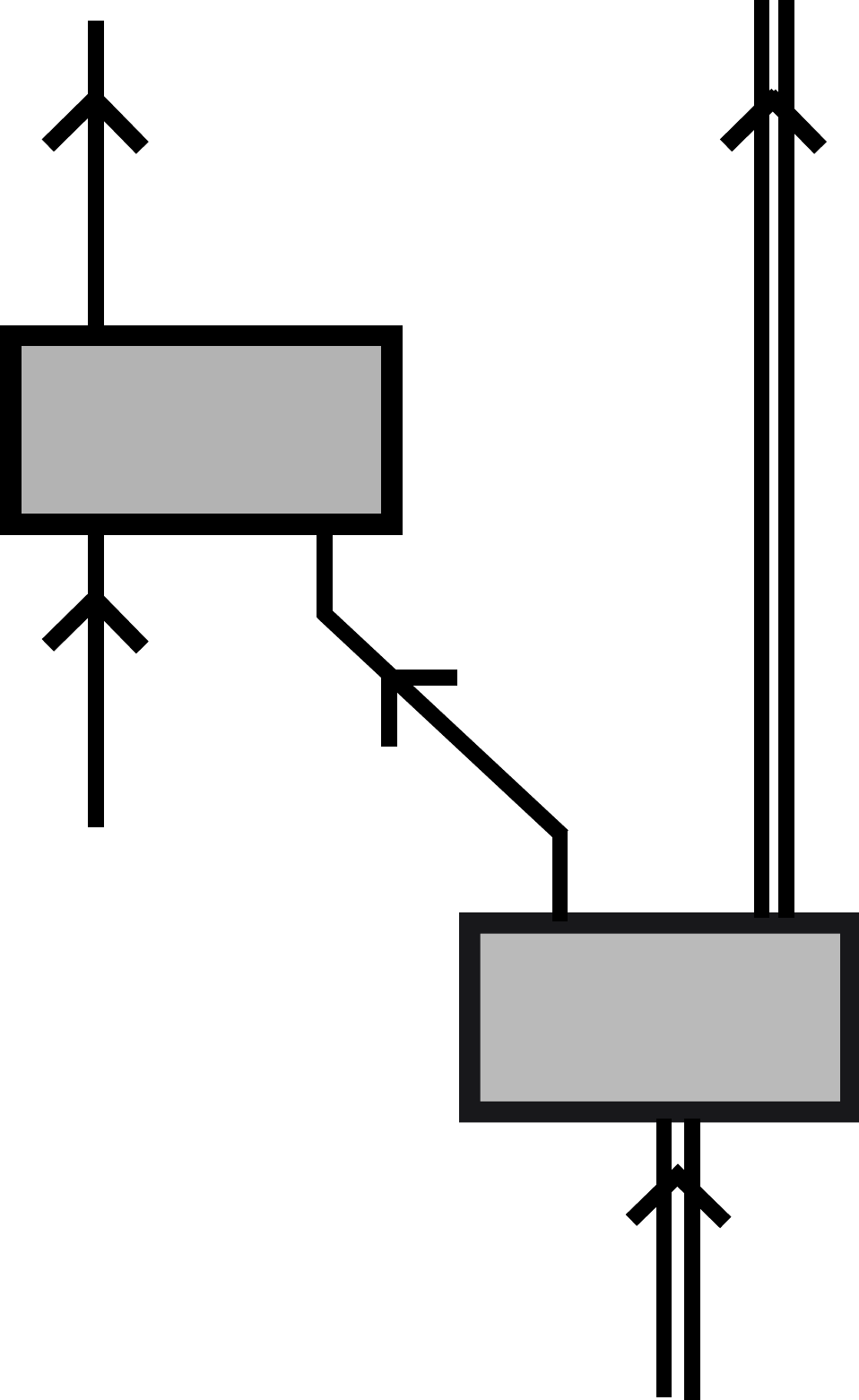}}
\put(-7,-25){$(a)$}
\put(-50,55){$x$}
\put(-50,95){$a$}
\put(3,60){$\lambda$}
\put(35,95){$\widetilde{\I}_{a|x}(\cdot)$}
}
\hspace{50mm} 
  \subcaptionbox{\label{fig:I2}}
{\put(-35,0){\includegraphics[width=0.22\textwidth]{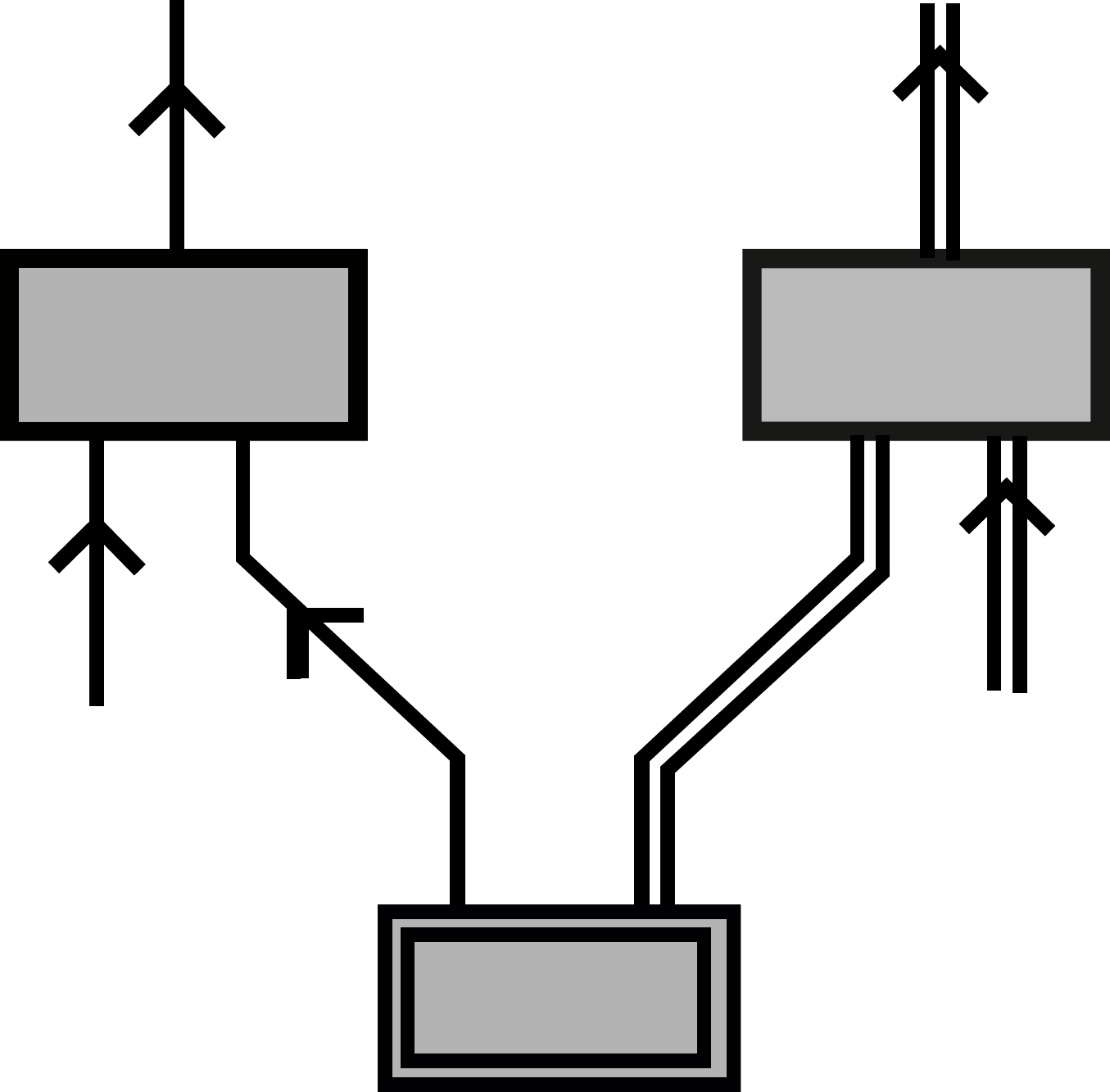}}
\put(7,-25){$(b)$}
\put(-45,40){$x$}
\put(-35,90){$a$}
\put(-10,20){$\lambda$}
\put(32,20){$\rho_\lambda$}
\put(60,95){$\widetilde{\I}^{NS}_{a|x}(\cdot)$}
}
\hspace{50mm}
  \subcaptionbox{\label{fig:I3}}
{\put(-35,0){\includegraphics[width=0.22\textwidth]{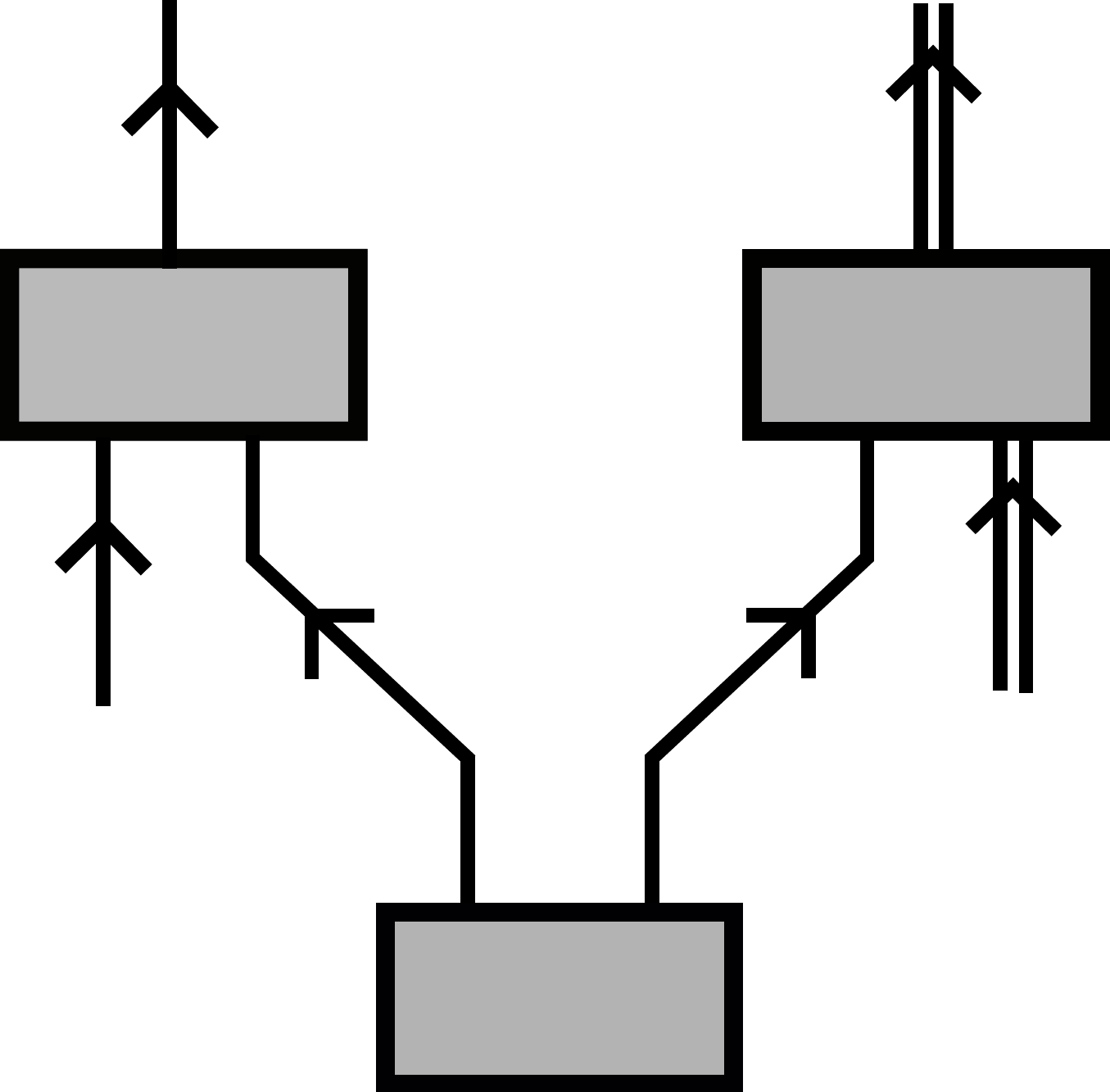}}
\put(7,-25){$(c)$}
\put(-45,40){$x$}
\put(-35,90){$a$}
\put(11,20){$\lambda$}
\put(60,95){$\I_{a|x}(\cdot)$}
}

  \end{center}
  \vspace*{-10mm}
  \caption{(a) An LOCC-free channel assemblage that allows signalling from Bob to Alice. (b) An LOCC-free  channel assemblage that is moreover non-signalling. (c) An LOSR-free channel assemblage. We prove here that the class of resources defined by (b) is identical to that defined by (c). 
  }\label{fig:I}
\end{figure}

\end{document}